\documentclass[11pt]{article}

\usepackage{hyperref}
\usepackage{amsthm,amsmath,amsfonts,amssymb,graphicx,mathtools,enumerate,placeins}
\usepackage{bm}
\usepackage{makecell}

\usepackage[normalem]{ulem} 
\usepackage{xspace}
\newcommand{\calstar}{Cutoff Calibration Error\xspace}
\newcommand{\calstarm}{\Delta_{\operatorname{Cutoff}}}
\newcommand{\calstarhat}{\widehat{\Delta}_{\operatorname{Cutoff}}}

\newcommand{\ece}{ECE\xspace}
\newcommand{\ecem}{\Delta_{\operatorname{ECE}}}

\newcommand{\dce}{dCE\xspace}
\newcommand{\dcem}{\Delta_{\operatorname{dCE}}}

\newcommand{\wcem}{\Delta_{\operatorname{wCE}}}

\usepackage{xcolor}

\usepackage[capitalise]{cleveref}
\let\cref\Cref
\usepackage[english]{babel}

\usepackage{booktabs}
\usepackage{nicematrix}

\newcommand{\ind}[1]{{{\mathbf{1}_{\left\{#1\right\}}}}}

\newcommand{\cX}{\mathcal{X}}

\newcommand{\EE}{\mathbb{E}}
\newcommand{\PP}{\mathbb{P}}

\usepackage{graphicx}

\usepackage{bbm}
\usepackage{tikz}
\usetikzlibrary{decorations.pathreplacing}
\usetikzlibrary{arrows.meta}

\usepackage[round]{natbib}

\usepackage[letterpaper, left=1.2truein, right=1.2truein, top = 1.2truein, bottom = 1.2truein]{geometry}

\theoremstyle{plain}

\newtheorem{theorem}{Theorem}[section]
\newtheorem{lemma}[theorem]{Lemma}

\newtheorem{proposition}{Proposition}[section]

\newtheorem{definition}{Definition}[section]
\newtheorem{corollary}{Corollary}[section]

\newtheorem{example}{Example}[section]

\makeatletter
\renewenvironment{proof}[1][\relax]{\par
  \pushQED{\qed}%
  \normalfont \topsep6\p@\@plus6\p@\relax
  \trivlist
  \item[\hskip\labelsep\itshape
    \ifx#1\relax \proofname\else\proofname{} of #1\fi\@addpunct{.}]\ignorespaces
}{%
  \popQED\endtrivlist\@endpefalse
}
\makeatother

\usepackage[blocks, affil-it]{authblk}

\title{Can a calibration metric be both testable and actionable?}
\author[1]{Raphael Rossellini}
\author[1]{Jake A. Soloff}
\author[1]{Rina Foygel Barber}
\author[2]{Zhimei Ren}
\author[1,3,4]{Rebecca Willett}
\affil[1]{Department of Statistics, University of Chicago} 
\affil[2]{Department of Statistics and Data Science, University of Pennsylvania}
\affil[3]{Department of Computer Science, University of Chicago} 
\affil[4]{NSF-Simons National Institute for Theory and Mathematics in Biology}

\begin{document}
\maketitle

\begin{abstract}
Forecast probabilities often serve as critical inputs for binary decision making. In such settings, calibration---ensuring forecasted probabilities match empirical frequencies---is essential. Although the common notion of Expected Calibration Error (ECE) provides actionable insights for decision making, it is not testable: it cannot be empirically estimated in many practical cases. Conversely, the recently proposed Distance from Calibration (dCE) is testable, but it is not actionable since it lacks decision-theoretic guarantees needed for high-stakes applications. To resolve this question, we consider Cutoff Calibration Error, a calibration measure that bridges this gap by assessing calibration over intervals of forecasted probabilities. We show that Cutoff Calibration Error is both testable and actionable, and we examine its implications for popular post-hoc calibration methods, such as isotonic regression and Platt scaling.
\end{abstract}

\section{Introduction}

To what extent should a decision-maker trust probability forecasts? For example, airlines may need to decide whether to cancel flights based on a weather model's predicted probability of severe weather. The reliability of such decisions fundamentally depends on whether these probabilities are \emph{calibrated}---that is, whether an 80\% forecast actually corresponds to an 80\% chance of the event occurring. Formally, a forecast model $f: \mathcal{X} \rightarrow [0,1]$ is perfectly calibrated~if
\begin{equation}\label{eq:perfect-calibration} 
    \mathbb{E}[Y\mid f(X)] = f(X) \;\text{ almost surely.}
\end{equation} However, many modern forecasting models, including neural networks, are known to produce miscalibrated probabilities \citep{guo2017calibration}. 

A challenge lies in quantifying how close the model $f$ comes to the ideal of perfect calibration~\eqref{eq:perfect-calibration}. There are two main desiderata for any measure of approximate calibration $\Delta(f)$. First, it should be \emph{testable}, in the sense that we should be able to reliably estimate the metric from finite data. In particular, we need to be able to test whether $f$ is approximately calibrated, $\Delta(f)\approx 0$. Second, any calibration metric should be \emph{actionable}---if $\Delta(f)$ is small, this should imply meaningful guarantees for downstream decision-making. For example, if a medical diagnostic system's probabilities have low calibration error, doctors should be confident that acting on its risk predictions according to clinical guidelines will lead to good patient outcomes. Unfortunately, as we will see, these two requirements are in tension: natural measures of calibration that provide strong decision-theoretic guarantees often turn out to be statistically intractable to estimate, and vice versa. We discuss in more detail these desiderata of actionability and testability in \cref{sec:decision_theory} and \cref{sec:testing}, respectively.

A popular calibration metric is the Expected Calibration Error (\ece):
\begin{equation*}
        \ecem(f):= \mathbb{E}\left[\Big|\mathbb{E}\left[Y \middle| f(X)\right]-f(X)\Big|\right].
\end{equation*}
Recent work has observed that \ece is actionable for binary decision problems \citep{kleinberg2023u,hu2024calibrationerrordecisionmaking}. However, \ece has severe limitations around its testability. Standard approaches to estimating $\ecem(f)$ are known to be very unstable \citep{kumar2019verified,nixon2019measuring} and biased \citep{roelofs2022mitigating}. These issues are ultimately not a technical limitation of existing approaches, but rather a consequence fundamental hardness results for estimating and constructing upper confidence bounds for $\ecem(f)$ \citep{gupta2020distribution}. 

In part to address the limitations of \ece, 
\citet{blasiok2023unifying} propose Distance from Calibration (\dce), which measures how far $f$ is from the nearest perfectly calibrated model:
\begin{equation*}
        \dcem(f) := \inf_{\substack{g: \mathcal{X} \rightarrow [0,1]\\ \mathbb{E}[Y \mid g(X)]=g(X)}}  \mathbb{E}\left[|g(X)-f(X)|\right].
\end{equation*} 
\dce is an easier notion of calibration to satisfy, as compared to \ece: indeed, for any $f$, $\dcem(f)\leq \ecem(f)$ \citep[Lemma 4.7]{blasiok2023unifying}.
Moreover, \dce solves the issue of testability, since unlike \ece, bounds on it can be estimated from data. However, this statistical advantage comes at a cost---\dce fails to provide the decision-theoretic assurances that made \ece valuable for applications. We show that models with small \dce can still lead to poor decisions.

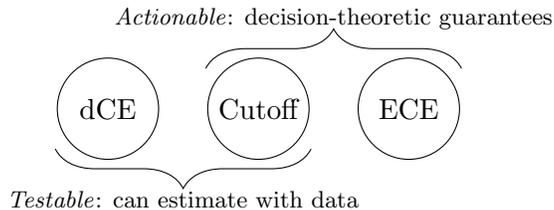
\begin{figure}[t!]
    \centering
    \begin{tikzpicture}

\node (ECE) at (4, 0) [draw, circle, minimum size=3.5em, inner sep=0] {ECE};
\node (iCE) at (2, 0) [draw, circle, minimum size=3.5em, inner sep=0] {Cutoff};
\node (dCE) at (0, 0) [draw, circle, minimum size=3.5em, inner sep=0] {dCE};

\draw [decorate,decoration={brace,amplitude=15pt,mirror,raise=1pt},yshift=0pt]
(-0.7,-0.5) -- (2.7,-0.5) node [black,midway,yshift=-20pt] 
{\footnotesize \emph{Testable}: can estimate with data};

\draw [decorate,decoration={brace,amplitude=15pt,raise=1pt},yshift=0pt]
(1.3,0.5) -- (4.7,0.5) node [black,midway,yshift=20pt] 
{\footnotesize \emph{Actionable}: decision-theoretic guarantees};

\end{tikzpicture}
    \caption{Comparing \calstar to ECE and dCE.}
    \label{fig:bubble_comparison}
\end{figure}

To bridge the gap between testable and actionable calibration measures, we consider the \emph{\calstar}, which we can view as an intermediate definition lying between \dce and \ece. This measure assesses calibration over ranges of predicted probabilities rather than at specific values. We prove that \calstar can be efficiently estimated from data at the parametric rate, while still providing guarantees about decision quality when using the calibrated forecasts. The key insight is that most practical decisions depend on whether a probability exceeds some threshold, making interval-based calibration particularly relevant.

Our key contributions may be summarized as follows.

\paragraph{Popular existing measures face limitations.} Building on the results of \citet{gupta2020distribution}, we establish that testing ECE is fundamentally hard by establishing that any calibration method with asymptotically vanishing ECE must have small effective support size. On the other hand, we demonstrate through counterexamples that dCE cannot provide the same decision-theoretic guarantees as ECE. In other words, we will see that ECE is not testable, while dCE is not actionable.

\paragraph{\calstar is testable and actionable:} We define \calstar and demonstrate that it combines the strongest benefits of ECE and dCE. First, controlling \calstar implies decision-theoretic guarantees in the binary decision setting under an intuitive monotonicity constraint. Second, like dCE, \calstar can be efficiently estimated from data using a plug-in approach that's closely related to a proposal of \citet{arrieta2022metrics}. We study relationships between \calstar, \ece and \dce in detail.\footnote{After submission of this paper, we learned of two concurrent works \citep{okoroafor2025near,qiao2025truthfulness} that address similar questions. Each proposes related definitions of calibration, and each considers the benefits of these definitions from the perspectives of a decision-theoretic framework, albeit in distinct ways; we will comment on these similarities and distinctions as appropriate when presenting our results below. In addition, \citet{okoroafor2025near} also considers testability, although our paper focuses more on how testability hardness results motivate the construction of \calstar.}

\paragraph{Implications for Platt scaling and isotonic regression:} We study popular post-hoc calibration methods of forecast probabilities in a distribution-free context. We provide a finite-sample bound on \calstar for isotonic regression. By contrast, we exhibit a counterexample for which Platt Scaling does not control \calstar even asymptotically.

\section{Our proposal: Cutoff calibration\label{sec:proposal}}

Consider a medical diagnostic system that models probabilities of various conditions. While existing calibration measures like \ece focus on the accuracy of specific probability values (e.g., historical outcomes when the model output was exactly 73\%), clinicians typically work with standardized risk thresholds---low, medium, and high risk categories that map to different intervention protocols. A doctor rarely needs to distinguish between a 73\% and 74\% risk, but they need to know whether a patient falls above or below key decision thresholds; perhaps the 80\% mark may prompt immediate attention.

This observation---that real-world decisions often depend on probability ranges rather than exact values---motivates us to reframe approximate calibration. Instead of assessing calibration separately at each possible value of $f(X)$, we evaluate calibration over intervals that correspond to potential decision thresholds:

\begin{definition}
\label{def:cutoff}
    \textbf{\calstar} is defined as follows:
    \begin{equation*}
        \calstarm(f) := \sup_{\textnormal{ intervals}~ I} \Big|\mathbb E[(Y-f(X))\ind{f(X) \in I}] \Big|.
    \end{equation*}
\end{definition}

This definition will be easier to work with than ECE, since we do not have to deal directly with the conditional expectation $\mathbb E[Y\mid f(X)]$. For a given interval $I$, the error $|\mathbb E[(Y-f(X))\ind{f(X) \in I}]|$ may be small either because the forecast is unlikely to take values in $I$ or because $f$ is well-calibrated on $I$. Concretely, if $\calstarm(f)\le \delta$, then for any interval $I$, 
\[
\Big| \, \mathbb E\left[ \ \mathbb{E}[Y\mid f(X)]-f(X) \,\big|\, {f(X) \in I} \ \right]\, \Big| \le \frac{\delta}{\mathbb{P}\{f(X) \in I\}}.
\]
Since the difficulty of the calibration problem adapts to the probability on each interval, \calstar is closely related to weighted calibration error \citep{gopalan2022low}. We further explore this connection below.

By searching over intervals $I \subseteq [0,1]$, this definition will be relevant for the binary decision framework, where decision rules take the form $\ind{f(X) \geq \tau}$, corresponding to $I = [\tau, 1]$. Taking the supremum over \emph{all} intervals allows downstream decision-makers to apply the threshold most relevant to their loss function.

How does \calstar relate to the existing calibration metrics, \ece and \dce? We establish a precise relationship through the following chain of inequalities.

\begin{proposition}\label{prop:ineqs}
For any $f: [0,1] \to [0,1]$, 
    \begin{equation*}
         \dcem(f)^2 / 9 \leq \calstarm(f) \leq \ecem(f). 
    \end{equation*}

\end{proposition}
In other words, \ece is a stronger notion of calibration than \calstar, and \calstar is a stronger notion of calibration than \dce. These comparisons are strict: as we will see in examples throughout the paper (see also \cref{apdx:examples_tightness}), it is possible to construct a function $f$ for which $\dcem(f)$ is arbitrarily close to zero while $\calstarm(f)$ stays bounded away from zero; or similarly, for which $\calstarm(f)$ is arbitrarily close to zero while $\ecem(f)$ stays bounded away from zero. In the next section, we connect each of these metrics to different kinds of weighted calibration error, which further clarifies their relationship by making the definitions more immediately comparable.

\subsection{Connections to weighted calibration error\label{sec:connections_wce}}

Weighted calibration error refers to a general family of calibration metrics introduced by \citet{gopalan2022low}. The definition is closely related to multicalibration \citep{hebert2018multicalibration} and multiaccuracy \citep{kim2019multiaccuracy} in the fairness literature, where the main distinction is that we only weight (or group) different values of $X$ via the predicted probability $f(X)$.

\begin{definition}
    Let $\mathcal{C} \subseteq [-1,1]^{f(\mathcal{X})}$. Then, the \textbf{weighted calibration error} of $f$ with respect to the class~$\mathcal{C}$ is defined as the following:
    \begin{equation*}
\Delta_{\operatorname{wCE}}(f; \mathcal{C}) := \sup_{w \in \mathcal{C}} \mathbb E [w(f(X))(Y - f(X))].
\end{equation*}
\end{definition}

\begin{figure}[t!]
    \centering
    \begin{tikzpicture}[
    >=stealth,
    every node/.style={font=\small},
    node distance=3cm
]

\tikzset{
  doubleimplies/.style={
    double,
    double distance=1.5pt,     
    line width=0.8pt,          
    -{Implies},
    shorten >=2pt,
    shorten <=2pt
  }
}

\tikzset{
  doubleequiv/.style={
    double,
    double distance=1.5pt,     
    line width=0.8pt,          
    <->,                       
    >=Implies,  
    shorten >=2pt,            
    shorten <=2pt
  }
}

\node[align=left] at (-5,2.) {\bfseries Weighted calibration \\
\bfseries w.r.t. a function class};

\node[align=center] (Lipsch) at (-1.9,2.)      {Lipschitz\\ functions};
\node[align=center] (BddVar) at (2.5,2.)      {Bounded-variation\\ functions};
\node[align=center] (AllMeas) at (6.5,2.)     {All bounded \\ functions};

\node at (0.1,2.)  {\Large $\subseteq$};
\node at (4.9,2.)  {\Large $\subseteq$};

\node[align=left] at (-5.5,0) {\bfseries Definition \\ \bfseries of calibration};

\node[align=center] (dCE) at (-1.9,0)  {dCE\\
\citep{blasiok2023unifying}};
\node[align=center] (Cal) at (2.5,0)  {\calstar \\(Our focus)};
\node[align=center] (ECE) at (6.5,0)  {ECE\\ \citep{dawid1982well}};

\draw[doubleequiv] (Lipsch) -- (dCE);

\draw[doubleequiv] (BddVar) -- (Cal);

\draw[doubleequiv] (AllMeas) -- (ECE);

\draw[doubleimplies] (ECE) -- (Cal);

\draw[doubleimplies] (Cal) -- (dCE);

\draw[doubleimplies, dashed]
    (dCE) 
      to[out=-20, in=-160]
    (ECE);

\node[draw, dashed, align=center, fill=white] at (2.5,-1.3) {assuming\\ finite support};





\end{tikzpicture}
    \caption{Relating \calstar to ECE and dCE. Arrows---of the form $\Delta_1\Rightarrow \Delta_2$---signify that a calibration metric $\Delta_2$ can be bound in terms of $\Delta_1$.}
    \label{fig:table_comparison}
\end{figure}
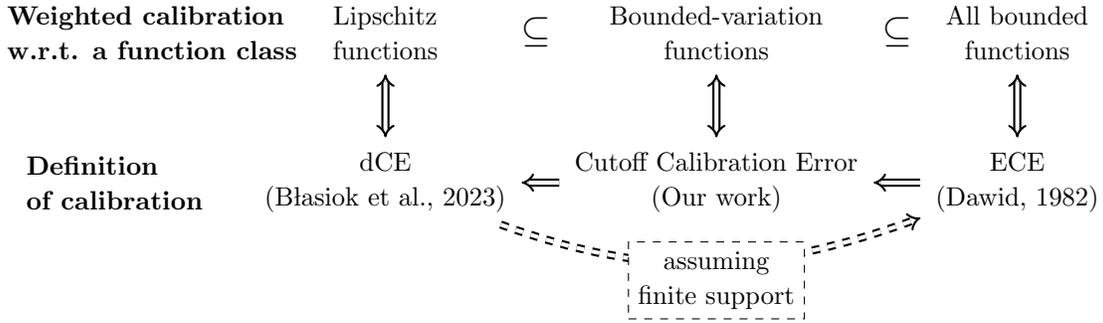

The various calibration metrics we have considered thus far---\dce, \ece, and \calstar---all have close connections to weighted calibration error for different choices of the function class~$\mathcal{C}$. For example, $\dcem(f)$ can be thought of as a $1$-Wasserstein distance between $f$ and the class of all perfectly calibrated functions. \citet{blasiok2023unifying} thus note that, by Kantorovich-Rubinstein duality, $\dcem(f)$ is polynomially equivalent to $\Delta_{\operatorname{wCE}}(f; \mathcal{L}_1)$ where $\mathcal{L}_1$ is the set of all $1$-Lipschitz weighting functions $w : [0,1]\to [-1, 1]$:
\begin{equation*}
    \dcem(f)^2 / 8 \leq \Delta_{\operatorname{wCE}}(f; \mathcal{L}_1) \leq 2 \cdot \dcem(f).
\end{equation*}
This form of weighted calibration error was previously known as \emph{weak calibration} \citep{kakade2008deterministic} and \emph{smooth calibration} \citep{gopalan2022low}. The connection to $1$-Lipschitz weighting functions lends further insight into why \dce may not be the most useful metric for decision making: since indicator functions are not Lipschitz, \dce is not compatible with hard thresholding rules. We further explore this issue in \cref{sec:decision_theory}.

On the opposite end of the spectrum, $\ecem(f)$ is exactly equal to the weighted calibration error using all bounded weighting functions $\Delta_{\operatorname{wCE}}(f; [-1,1]^{[0,1]})$. This observation suggests a challenge of working with \ece: the corresponding function class is too large to be able to estimate this quantity in general. We further explore this issue in \cref{sec:ece_hardness}.

Finally, note that $\calstarm(f)$ is already in the form of a weighted calibration error, where the weighting functions are indicator functions over intervals. Since the objective $$\mathbb E [w(f(X))(Y - f(X))]$$ in weighted calibration error is linear in $w$, we can replace this class with its convex hull, leading to a much richer class.

\begin{proposition}    \label{prop:bv_fns}
Let $\mathcal{B}_M$ be the class of all functions on $[0,1] \rightarrow [-1,1]$ whose total variation is at most $M \in \mathbb R^+$. Then, for $M \geq 2$,
    \begin{equation*}
        \calstarm(f) \leq \wcem(f; \mathcal{B}_M) \leq (M+2) \calstarm(f).
    \end{equation*}
\end{proposition}

This result offers a concrete sense in which \calstar achieves a useful middle ground between \dce and \ece. The class of bounded variation functions has relatively low complexity, but it is sufficiently expressive to include indicator functions over intervals.

\section{Actionable calibration: Implications for binary decision theory
\label{sec:decision_theory}}

We now turn to analyzing calibration measures through the lens of decision theory, examining what guarantees they provide about forecast quality in practice. In this section, we focus on the binary decision setting, i.e., $Y\in \{0,1\}$. In \cref{apdx:actionable_extensions}, we demonstrate  that analogous results hold for a sign-testing loss function when $Y \in [0,1]$, and even all bounded proper scoring rules when $Y \in \{0,1\}$, under mild regularity assumptions.

\subsection{The simple binary decision setting}

Consider a simple but fundamental scenario where a decision-maker must choose between two actions (like canceling or proceeding with a flight) based on a probability forecast $f(X)$ of some $Y\in \{0,1\}$. The decision-maker's loss function is determined by their relative tolerance for false positives and false negatives. We write their loss function as $\ell_{\mathrm{bd}}: \{0,1\} \times \{0,1\} \times [0,1] \rightarrow [0,1]$, where ``bd''  refers to ``binary decision'' loss. For some $\tau\in [0,1]$, the loss is
\begin{equation*}
    \ell_\mathrm{bd}(Y, \widehat{Y}; \tau) = \tau (1-Y)\widehat{Y} + (1-\tau) Y(1-\widehat{Y}),
\end{equation*}
where $\widehat{Y} \in \{0,1 \}$ encodes our preferred action if we knew $Y$. 

Given an observed forecast $f(X)$, the Bayes decision rule is given by $\ind{\mathbb E[Y\mid f(X)] \geq \tau}$. That is to say, the expected loss is minimized when the resources are deployed ($\widehat{Y} = 1$) if and only if the calibrated forecast probability is above $\tau$. In practice, $\mathbb E [Y \mid f(X)]$ is unknown, so decision-makers naturally resort to the following plug-in decision rule:
\[
    \widehat{Y} = \ind{f(X) \geq \tau}.
\]
But, if $f$ is highly miscalibrated, this might be a poor choice of decision rule.

Examining this example, we can see that we are asking whether the decision rule $\ind{f(X)\geq \tau}$ is (nearly) as good as some other decision rule obtained by transforming $f$, $\ind{h(f(X))\geq \tau}$ (observe that the Bayes rule can be written in this form by taking $h(t) = \mathbb{E}[Y\mid f(X)=t]$). We are particularly interested in whether this favorable property is ensured for functions $f$ that satisfy $\Delta(f)\approx 0$, given some particular choice of calibration measure $\Delta$. This is what we mean by saying that $\Delta$ is \emph{actionable}: we are asking whether low miscalibration ($\Delta(f)\approx 0$) implies that the plug-in decision rule is nearly optimal relative to some class of possible modifications of the rule---that is, some set of transformations $\mathcal H$.

\subsection{Comparing calibration measures through decision guarantees}

We now show that the three calibration measures exhibit strikingly different relationships with decision quality. At a high level, Distance from Calibration (\dce) can be arbitrarily small even when a forecast leads to substantially suboptimal decisions. Expected Calibration Error (\ece) provides the strongest guarantee: it bounds the difference between achieved and optimal risk under any transformation of the forecast. Cutoff Calibration Error strikes an appealing middle ground: while not matching the fully general guarantee for \ece, it ensures near-optimal decisions when we restrict attention to the natural class of monotonic decision rules.

Below, we present an example showing that dCE is not guaranteed to be actionable in the binary decision setting, following a similar result from \citet{hu2024calibrationerrordecisionmaking}. Notably, our example is adapted from an example used in \citet{blasiok2023unifying} to motivate their definition of \dce, since they find it undesirable for a calibration metric to change dramatically under small perturbations in the forecasts. Our example warns that continuity comes at the expense of decision-theoretic guarantees.

Before presenting the example, we first define a measure of risk for the plug-in decision rule: 
\begin{equation}
    R_\mathrm{bd}(f; \tau) := \mathbb E[\ell_\mathrm{bd}(Y, \ind{f(X) \geq \tau}; \tau)].\label{eq:risk_bd}
\end{equation}

\begin{example} Suppose $X \sim \mathrm{Bernoulli}(0.75)$ and $Y\mid X \equiv X$. Observe that the constant function $f(x) := 0.75$ is perfectly calibrated, so $\dcem(f)= \ecem(f) = 0$. On the other hand, for a slightly perturbed function $\Tilde{f}(x) := 0.75 + \epsilon - 2\epsilon x$, we have $\dcem(\Tilde{f}) \leq \epsilon$ and $\ecem(\Tilde{f}) = \frac{3}{8}+\epsilon$. In this case, \dce remains small for small values of $\epsilon$, but \ece is bounded away from zero.

Suppose $\tau=0.75$ and we only have access to $\Tilde{f}(X)$. Then, the Bayes decision rule will be to use $\widehat{Y} = \ind{\mathbb E [Y\mid \Tilde{f}(X)] \geq 0.75}$. If we use $\ind{\Tilde{f}(X) \geq 0.75}$ in place of the Bayes decision rule, then the risk $R_\mathrm{bd}(\Tilde{f}(X); \tau)] = \frac{3}{8}$. \emph{However}, the Bayes decision rule actually has 0 risk, since in fact $Y = \ind{\mathbb E [Y\mid \Tilde{f}(X)] \geq 0.75}$. Therefore, while $\dcem(\Tilde{f})$ can be arbitrarily small, the discrepancy in losses between the plug-in decision rule and the oracular Bayes decision rule remains bounded away from 0.    
\end{example}
From this example, we can see that a small dCE error, $\dcem(f)\approx 0$, does not necessarily ensure that the forecast $f$ is actionable.\footnote{Interestingly, there is a certain sense in which dCE is ``close'' to being actionable: adding a small amount of noise to a forecast $f(X)$ with small dCE will lead to an actionable rule. In fact, though, this is because $f(X)+\mathrm{noise}$ is then guaranteed to have small ECE \citep{blasiok2024smooth}. In other words, it is possible to build a low-ECE (and therefore, actionable) forecast by using the low-dCE rule $f(X)$ as a starting point---but $f(X)$ is not itself actionable, and decisions made by thresholding $f(X)+\mathrm{noise}$ might be arbitrarily different from decisions made using $f(X)$.}

In contrast, when ECE is small, we can directly guarantee that the risk associated with $f$ is close to the best possible wrapper $h \circ f$. We often refer to this difference as the ``\emph{risk gap}.''

\begin{proposition}\label{prop:decision_theory_ece} For any $\tau \in [0,1]$, 
    \begin{equation*}
        R_\mathrm{bd}(f; \tau) - \inf_{h: [0,1]\rightarrow[0,1]} R_\mathrm{bd}(h \circ f; \tau) \leq \ecem(f).
    \end{equation*}
\end{proposition}
The proof of this follows immediately after the following rearrangement.
\begin{lemma}
    Let $h:[0,1]\rightarrow [0,1]$ be arbitrary. Then,
\begin{equation}
\begin{aligned}
    R_\mathrm{bd}(f; \tau) &- R_\mathrm{bd}(h \circ f; \tau) \\
    &=  \mathbb E [(Y - \tau)\ind{f(X)<\tau}\ind{h(f(X))\geq \tau}] + \mathbb E [(\tau - Y)\ind{f(X)\geq\tau}\ind{h(f(X))<\tau}].
\end{aligned}\label{eq:decision_theory_rearrangement}
\end{equation}\label{lemma:decision_theory_rearrangement}
\end{lemma}

\cref{prop:decision_theory_ece} also follows from a more general result in \citet{hu2024calibrationerrordecisionmaking}, where instead of considering our specific loss function they worked with all bounded loss functions that are proper scoring rules. They propose a measure of calibration that bounds this risk gap (often termed ``swap regret'' in game theory). Their proposed definition of calibration ends up resembling the right hand side of \eqref{eq:decision_theory_rearrangement}, with a supremum taken over $h$ and $\tau$, using steps found also in \citet{kleinberg2023u, li2022optimization}. While their proposed measure of calibration is $\leq \ecem(f)$, it is also $\geq \ecem(f)^2$ \citep[Theorem 4.1]{hu2024calibrationerrordecisionmaking}. Unfortunately, in the setting of a continuous random variable $f(X)$ that is the focus of our work, this means that the proposed calibration measure of \citet{hu2024calibrationerrordecisionmaking} is not testable, like ECE. 

In parallel, \citet{kleinberg2023u} leverage similar observations and propose a calibration metric that is in fact testable. However, its guarantee is for external regret, which compares against a non-personalized treatment policy. Thus, the metric proposed in \citet{kleinberg2023u} is not actionable in the sense that we define here in our work. Relatedly, \citet{qiao2025truthfulness} consider decision-theoretic guarantees from a different perspective, including benchmarking against non-personalized treatment policies.

In reality, the best possible wrapper of $f(x)$ (namely, $\mathbb E[Y\mid f(X) = f(x)]$) can be very non-smooth, especially when $f(X)$ follows a continuous distribution. The lack of smoothness is not just an abstract concern. For example, no one would implement a decision rule of only bringing your umbrella if the forecast $f(X)$ predicts between a 20\% and 50\% chance of rain. Yet, an oracle $h$ does not preclude the possibility of suggesting users implement a decision rule that defies common sense. Therefore, the risk gap studied in \cref{prop:decision_theory_ece} is benchmarking against an unrealistic standard.

A natural shape constraint on $h$ that comports with common sense is that it be \emph{monotonically non-decreasing}. By constraining ourselves to monotonic wrappers around $f$, we can upper bound the \emph{monotone} risk gap by \calstar.

\begin{proposition}\label{prop:decision_theory_calstar} For any $\tau \in [0,1]$,
    \begin{equation*}
        R_\mathrm{bd}(f; \tau) - \inf_{\substack{h: [0,1] \rightarrow [0,1]\\ \textnormal{monotone}}} R_\mathrm{bd}(h\circ f; \tau) \leq 2 \calstarm(f).
    \end{equation*}
\end{proposition}

An alternate way of proving this result, up to constant factors, is to use the results of \citet[Lemmas 3.2 and 3.5]{okoroafor2025near}, who examine this question from the perspective of omniprediction \citep{gopalan2021omnipredictors}. The salient fact in both our proof and one leveraging results in \citet{okoroafor2025near} is that pre-images of super-level sets of univariate monotone functions are half-spaces.

While we state \cref{prop:decision_theory_calstar} just for binary decision loss, a corollary is that we can bound the risk gaps for arbitrary bounded proper scoring rules when $Y \in \{0,1 \}$ (\cref{apdx:actionable_extensions}). This corollary follows from the Schervish representation \citep{schervish1989general} of proper scoring rules, which demonstrates that proper scoring rules are a mixture of binary decision loss across different $\tau$ values.

\subsection{Experiment\label{sec:experiment}}

\paragraph{Set-up} To demonstrate how \calstar satisfies our \emph{actionable} criterion while dCE does not, we consider the following simulation setting.\footnote{Relevant code is available at \href{https://github.com/rrross/CutoffCalibration}{https://github.com/rrross/CutoffCalibration}} Let $\alpha \in [0,1]$ be a parameter, and define 
\[
\mathbb E [Y\mid X] = \alpha (1 - 2X)^2 + (1-\alpha) X.
\]
This expression is a convex combination between a parabola (symmetric about 0.5) and the $y=x$ line. For each of our 100 simulation runs, we draw $\alpha \sim \mathrm{Uniform}(0,1)$ and then sample each $X_i \sim \mathrm{Uniform}(0,1)$ independently. Each $Y_i$ is drawn from a Bernoulli distribution with mean $\mathbb E[Y_i\mid X_i]$. We fit a univariate logistic regression $f: [0,1] \rightarrow [0,1]$ on $\{(X_i, Y_i)\}_{i=1}^n$, with $n = 500$. Notably, as $\alpha$ increases, the logistic model is increasingly misspecified, as $\mathbb E [Y\mid X]$ gains curvature. We evaluate our binary decisions using \eqref{eq:risk_bd}, with $\tau=0.35$.

\paragraph{Results} Figure~\ref{fig:simulation_results} shows how the risk gap and the monotone risk gap relate to $\calstarm(f)$, $\ecem(f)$, and $\dcem(f)$. We observe that small values of $\ecem(f)$ and $\calstarm(f)$ correlate with lower risk gaps, guaranteeing that there is no wrapper and monotonic wrapper (respectively) that would meaningfully reduce the expected loss in the binary decision problem. In contrast, $\dcem(f)$ has a noisy---and at times negligible---relationship with either gap. Thus, even when $\dcem(f)$ is small, there can exist a wrapper function that would have substantially reduced the expected loss of one's decision. Consequently, the decision-theoretic guarantees with \ece and \calstar (see \cref{prop:decision_theory_ece} and \cref{prop:decision_theory_calstar}) are supported by the results of this simulation 
setting, but dCE does not enjoy these guarantees, neither theoretically nor in the simulation results.

\begin{figure}[t!]
    \centering
    \includegraphics[width=\linewidth]{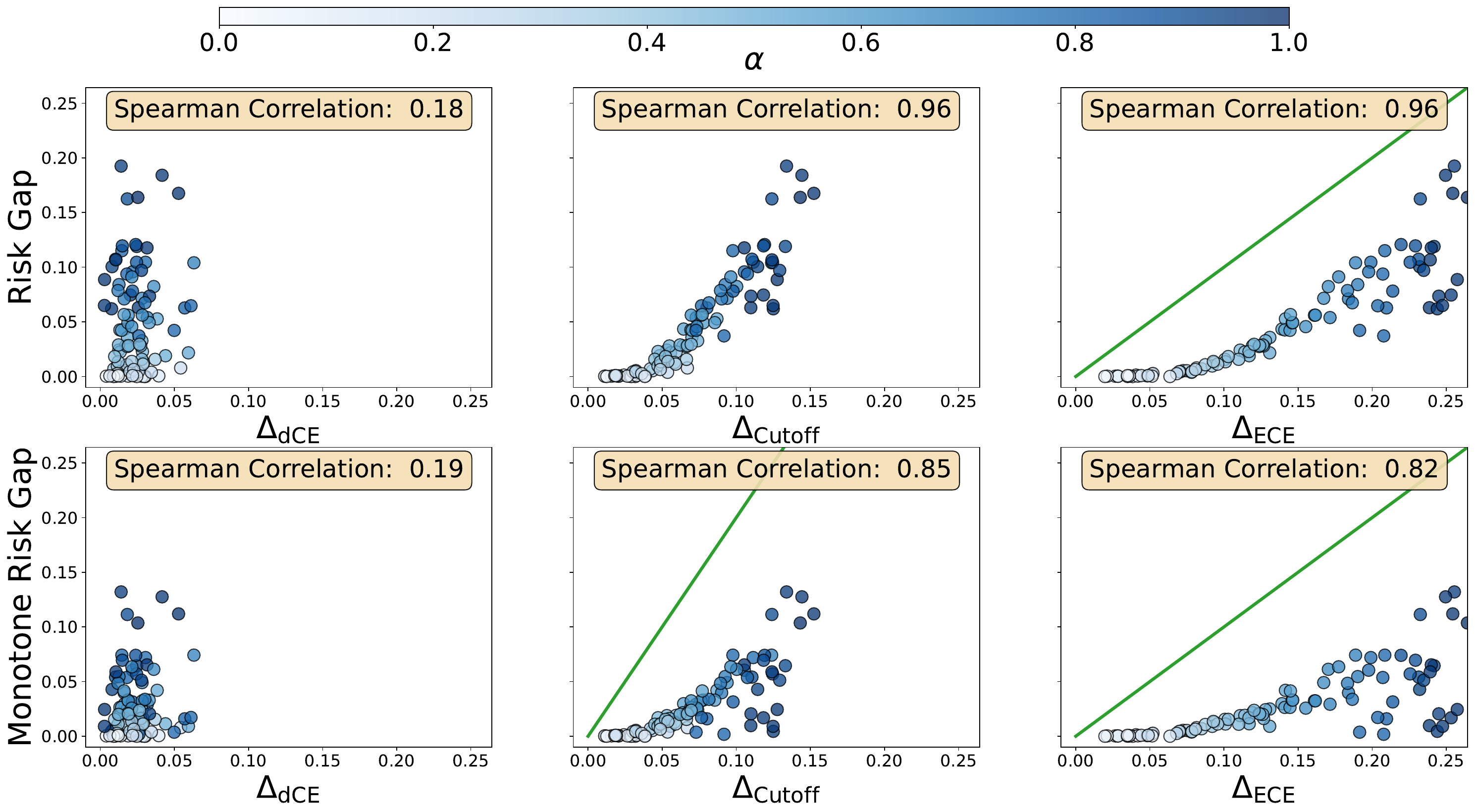}
    \caption{$\calstarm(f)$ and $\ecem(f)$ being small corresponds to smaller risk gaps (both standard and monotone). In contrast, when $\dcem(f)$ is small, there is no guarantee that the risk gaps will also be small. Green lines denote upper bounds guaranteed by \cref{prop:decision_theory_ece} and \cref{prop:decision_theory_calstar}. Larger $\alpha$ values correspond to more model-misspecification in our forecast model $f$. Calibration measures $\dcem(f), \calstarm(f), \ecem(f)$ are oracle quantities in the sense that we calculate them using knowledge of $\mathbb E [Y\mid X]$; see \cref{apdx:actionable_extensions}.}
    \label{fig:simulation_results}
\end{figure}

\section{Testable calibration: Distribution-free guarantees}\label{sec:testing}

In order for a calibration metric to be useful for decision-making in practice, as discussed in Section~\ref{sec:decision_theory}, we need to choose a metric $\Delta$ that is \emph{actionable}. But this is not sufficient: we also need to be able to verify empirically that $\Delta^P(f_n)\approx 0$,\footnote{When we write $\Delta^P(f_n)$, we proceed as if $f_n$ were fixed. We use the $P$ superscript to emphasize that the expectation is only taken over $P$, not over the data sampled from $P^n$ and used to train $f_n$ (that is, the expectation is computed conditional on $f_n$). In particular for a fixed function $f$, $\Delta(f)=\Delta^P(f)$, since conditioning on the fixed function $f$ would have no effect. 
} in order to know whether our particular trained model~$f_n$ can then be safely used for decision-making. In other words, for a small constant $c>0$,
\begin{quote}
    Given a sample of data points drawn i.i.d.\ from $P$, can we fit a function $f_n$ to the data, in such a way that guarantees $\Delta^P(f_n)\leq c$ with high probability?
\end{quote}
Of course, this is always possible with a trivial solution: we can simply return a constant function $f_{n,\textnormal{const}}(x) \equiv \hat\mu_Y$ where $\hat\mu_Y$ is the sample mean of $Y$ (since we can always estimate the mean $\EE [Y]$ with error $\mathcal{O}_P(n^{-1/2})$, this will ensure that $\Delta^P(f_{n,\textnormal{const}})$ is low). However, we would ideally want to use a state-of-the-art fitted model $f_n$, not some overly simple model such as a constant function. We can therefore ask a related question:
\begin{quote}
    (Testability.) Given a fixed function $f$ and a sample of data points drawn i.i.d.\ from $P$, can we test the hypothesis $\Delta(f)\leq c$?
\end{quote}
To see how this relates to the previous question, we can consider the following workflow, given a data set of size $n$:
\begin{itemize}
    \item Using $n/2$ data points, train an initial predictive model $f_{n,\textnormal{init}}$. 
    \item Using the remaining $n/2$ data points, test whether $\Delta^P(f_{n,\textnormal{init}})\leq c$. If yes, then return $f_n = f_{n,\textnormal{init}}$. If not, then return a constant function, $f_n = f_{n,\textnormal{const}}$, for $f_{n,\textnormal{const}}(x) \equiv \hat\mu_Y$. 
\end{itemize}
In other words, if $\Delta$ is testable, then we are able to return a model $f_n$ that is (with high probability) guaranteed to satisfy $\Delta^P(f_n)\leq c$. 

In this section, at a high level, our results show that \dce and \calstar are both testable---and therefore, it is possible to use data to produce a fitted model $f_n$ that is guaranteed to satisfy these calibration metrics. On the other hand, \ece is testable only for piecewise constant functions $f$ and is not testable more generally; equivalently, any procedure that returns a $f_n$ that is guaranteed to have low \ece, must therefore return a $f_n$ that is  (usually) piecewise constant, which immediately excludes many standard models such as logistic regression (or Platt scaling).

\subsection{Testability of \dce and \calstar}\label{sec:twostage}
We will next see that \calstar and \dce are both testable: it is possible to test the hypothesis $\calstarm(f)\leq c$ or $\dcem(f)\leq c$. In particular, as we will formalize below, we can construct practical procedures to return a function $f_n$ guaranteed to satisfy a bound on calibration error.

We begin by considering \calstar. First, we establish that the natural estimator for $\calstarm(f)$ is consistent.
\begin{proposition}\label{prop:estimation}
    Suppose $(X_i, Y_i) \overset{iid}{\sim} P $ for $i=1,\dots, n$ and $Y \in [0,1]$. 
    Let $\calstarhat(f): \mathcal{F} \rightarrow [0,1]$ be a plug-in estimator for \calstar defined as
    \begin{equation*}
        \calstarhat(f) = \sup_{I: \text{ interval}} \left|\frac{1}{n} \sum_{i=1}^n (Y_i-f(X_i))\ind{f(X_i) \in I}\right|.
    \end{equation*}
    Then, for any $\delta>0$, with probability at least $1-\delta$,
    \begin{equation*}
        \left|\calstarhat(f) - \calstarm(f) \right| 
        \leq \frac{20+\sqrt{2\log(1/\delta)}}{\sqrt{n}}.
    \end{equation*}
\end{proposition}

In particular, we can therefore use this estimator to test whether cutoff calibration error is low. In other words, given some threshold $c>0$, if $\calstarhat(f)\leq c - \frac{20+\sqrt{2\log(1/\delta)}}{\sqrt{n}}$ then we can certify (with $1-\delta$ confidence) that $\calstarm(f)\leq c$---and in particular, this procedure may have nontrivial power as soon as $n\gtrsim \frac{\log(1/\delta)}{c^2}$. An analogous result on testability for a closely related calibration measure is provided by \citet[Lemma C.2]{okoroafor2025near}.

More generally, however, we may be interested in providing a procedure that is \emph{guaranteed} to offer \calstar bounded by some threshold $c$, rather than \emph{testing} whether \calstar is bounded by $c$ (or not). We can also use the result above to provide a procedure that offers such a guarantee.
Given a data set $(X_1,Y_1),\dots,(X_n,Y_n) \overset{iid}{\sim}P$, a threshold $c>0$ and $\delta\in(0,1)$, and a model fitting algorithm $\mathcal{A}$ (e.g., a neural net),
\begin{itemize}
    \item Train a model $f_{n,\textnormal{init}}$ using the first half of the data:
    $f_{n,\textnormal{init}} = \mathcal{A}\big((X_1,Y_1),\dots,(X_{\lceil n/2\rceil},Y_{\lceil n/2\rceil})\big)$.
    \item Estimate the calibration error of $f_{n,\textnormal{init}} $ using the second half of the data:
    \[\calstarhat(f_{n,\textnormal{init}}) = \sup_{I: \text{ interval}} \left|\frac{1}{\lfloor n/2\rfloor} \sum_{i=\lceil n/2\rceil +1}^n (Y_i-f(X_i))\ind{f_{n,\textnormal{init}}(X_i) \in I}\right|.\]
    \item If $\calstarhat(f_{n,\textnormal{init}}) \leq c - \frac{20+\sqrt{2\log(1/\delta)}}{\sqrt{\lfloor n/2\rfloor }}$, return $f_n =f_{n,\textnormal{init}}$, else return $f_n = f_{n,\textnormal{const}}$ where
    \[f_{n,\textnormal{const}}(x)\equiv \frac{1}{\lfloor n/2\rfloor}\sum_{i=\lceil n/2\rceil +1}^n  Y_i.\]
\end{itemize}
We then have the following guarantee:
\begin{corollary}\label{cor:guarantee_calstar}
    For any distribution $P$, and any model fitting algorithm $\mathcal{A}$, the procedure described above satisfies
    \[\PP\{\calstarm^P(f_n)\leq c\} \geq 1-2\delta\]
    as long as $c\geq \sqrt{\frac{\log(1/\delta)}{2\lfloor n/2\rfloor}}$.
\end{corollary}
In particular, note that the result holds uniformly over any distribution $P$, and the procedure does not require any knowledge of $P$---that is, this is a distribution-free guarantee. Moreover, there are no restrictions on the model fitting algorithm $\mathcal{A}$, and so the returned model $f_n$ may be quite complex and adaptive to the data, if we choose our algorithm well for the task at hand.

Next, consider \dce. In fact, since $\dcem(f)\leq 3\sqrt{\calstarm(f)}$ by Proposition~\ref{prop:ineqs}, this means that any guarantee of an upper bound on $\calstarm(f)$ (as in Corollary~\ref{cor:guarantee_calstar} above, say) immediately yields a bound on $\dcem(f)$, as well. Relatedly, \citet[Section 9]{blasiok2023unifying} establish that a calibration measure related to \dce can be estimated consistently using an i.i.d.\ sample, with error $\mathcal{O}_P(n^{-1/2})$, which leads to lower and upper bounds on \dce, as well.

\subsection{Hardness result: \ece is not testable\label{sec:ece_hardness}}
In contrast to \dce and \calstar, the \ece calibration metric is not testable in the distribution-free setting; there is no estimator of $\ecem(f)$ that is guaranteed to be consistent for any $f$ and uniformly over any distribution $P$. In particular, this means that we cannot construct a two-stage procedure in the style of the procedure described in Section~\ref{sec:twostage} for \calstar and \dce. More strongly, we will now see that \emph{any} mechanism for returning a function $f_n$ guaranteed to have low \ece is necessarily forced to return an output that is usually piecewise constant; we cannot output, say, a continuous and strictly increasing $f_n$ (as would be the case for procedures such as post-hoc calibration via Platt scaling), if we require a distribution-free guarantee on ECE.

To state the result, we first need some additional notation. First, for any $\gamma\in[0,1]$ and any distribution $P$ on $\mathcal{X}\times[0,1]$ we define 
\[\sigma^2_\gamma(P) = \inf\left\{ \EE_P\left[\textnormal{Var}(Y\mid X) \cdot\ind{X\in A} \right] \ : \ A\subseteq\mathcal{X}, \PP_P(X\in A) \geq \gamma\right\}.\]
We can think of this quantity as an ``effective variability'' of $P$---in particular if $\textnormal{Var}(Y\mid X)\geq \sigma^2>0$ almost surely, then $\sigma^2_\gamma(P) \geq \sigma^2$ for any $\gamma$.
Next, we also define a notion of ``effective support size'':
\[S_\gamma(f,P) = \inf\left\{ k\geq 1 : \PP_P\big(f(X) \in \{t_1,\dots,t_k\}\big) \geq 1-\gamma \textnormal{ for some $t_1,\dots,t_k\in[0,1]$}\right\},\]
or $S_\gamma(f,P)=+\infty$ if the set is empty. If $f(X)$ takes at most $k$ many values (under $X\sim P$), then $S_\gamma(f,P)\leq k$ for any $\gamma$.\footnote{As for calibration error, when the function $f_n$ is data-dependent, we treat $f_n$ as fixed for the purpose of calculating $S_\gamma(f_n,P)$---that is, $\PP_P(f_n(X)\in\{t_1,\dots,t_k\})$ should be interpreted as a probability calculated with respect to an independent draw of $X$, when we condition on the trained model $f_n$.}

\begin{theorem}\label{thm:hardness_ece}
    Fix any $c,\delta>0$ and any sample size $n\geq 1$. Let     \begin{equation*}
        \mathcal{A}:(\mathcal{X} \times[0,1])^n \rightarrow\{\text { measurable functions } f: \mathcal{X} \rightarrow[0,1]\}
    \end{equation*}
    be any procedure that inputs a data set $(X_1,Y_1),\dots,(X_n,Y_n)$, and returns a fitted function $f_n = \mathcal{A}\big((X_1,Y_1),\dots,(X_n,Y_n)\big)$. 
    
    If $\mathcal{A}$ satisfies the distribution-free guarantee
    \[\PP_{P^n}\{\ecem^P(f_n)\leq c\} \geq 1-\delta\textnormal{ for all distributions $P$},\]
    then it holds that
    \[\PP_{P^n}\left\{S_\gamma(f_n,P) \leq n^2 \right\} \geq 1 - \frac{2e(c+\delta+n^{-1})}{\sigma^2_\gamma(P)}\textnormal{ for all distributions $P$}.\]
\end{theorem}

If $c$ and $\delta$ are small (i.e., $\mathcal{A}$ offers a meaningful guarantee on \ece),
then for any $P$ with an effective variance bounded away from zero, we must have $S_\gamma(f_n,P) \leq n^2$ with high probability---that is, $f_n$ must be a 
function that is piecewise constant on most of the domain, with effective support at most $\mathcal{O}(n^2)$.

We emphasize that this result applies to \emph{any} procedure $\mathcal{A}$---for instance, $\mathcal{A}$ can be the outcome of training a model on part of the data, and then using a post-hoc calibration procedure such as Platt scaling on an additional batch of data. In other words, this result implies that post-hoc calibration does not have any meaningful distribution-free guarantee in terms of the \ece, unless we use a procedure that returns a piecewise constant output. Thus, ECE does not satisfy our testability criterion for a very broad class of models.

Related work by \citet{gupta2020distribution} examines this type of question from an asymptotic perspective, finding that any post-hoc calibration procedure satisfying an asymptotic guarantee on \ece cannot return injective functions. In \cref{apdx:testability}, we will compare our result to this existing work and see how \citet{gupta2020distribution}'s asymptotic hardness result is implied by \cref{thm:hardness_ece}.

\paragraph{Estimating ECE via binning?}
In practice, it is common to estimate ECE with a binning-based approach: that is, given a function $f$ and a partition $[0,1] = A_1\cup\dots\cup A_N$ (with $N$ small relative to sample size $n$), we use the available data estimate the error $$|\mathbb E[Y \mid f(X)\in A_i] - \mathbb E[f(X) \mid f(X)\in A_i]|,$$ and aggregate over bins $i=1,\dots,N$, to estimate $\ecem(f)$ \citep{blasiok2023unifying}. In fact, this type of binned approximation (known as the binned ECE) is actually estimating the \ece of \emph{the binned approximation} to the function $f$ (that is, a function $g$ that is piecewise constant, taking the value $g(x) = \mathbb E[f(X) \mid f(X)\in A_i]$ for each $x\in A_i$)---and without further assumptions, it is possible to have $\ecem(g)\approx 0$ even when $\ecem(f)$ is large. This challenge is closely connected to our hardness result, Theorem~\ref{thm:hardness_ece}.

\section{Methodological insights on post-hoc calibration\label{sec:post_processing}}
In this section, we examine the implications of \calstar on post-hoc calibration. Post-hoc calibration, in short, aims to complement a pre-fit forecast model $f$ by using data to find an $h_n: [0,1] \rightarrow [0,1]$ such that $h_n \circ f$ achieves small calibration error.

\subsection{Isotonic regression enjoys uniform asymptotic control of \calstar}

A popular calibration method is to use isotonic regression on the observed outcomes~$Y$ as a monotone function of the forecasts~$f(X)$ \citep{zadrozny2002transforming}. The isotonic constraint is natural for reasons similar to those outlined in \cref{sec:decision_theory}. Namely, one may be concerned that a non-monotonic $h_n$ would be the result of over-fitting. Moreover, even if there were true non-monotonic regions in $\mathbb E[Y\mid f(X)]$ as a function of $f(X)$, one may view correcting those errors as more in the purview of creating a better base forecast $f$.

\begin{proposition}    \label{prop:iso}
    Suppose $(X_i, Y_i) \overset{iid}{\sim}P$ for $i=1,\dots, n$. Assume $Y_i \in [0,1]$ and $f$ is fixed. Let $h_n$ refer to running isotonic regression on $\{(f(X_i), Y_i) \}_{i=1}^n$. Fix $\delta>0$. Then,
    \begin{equation*}
    \calstarm^P(h_n \circ f) \leq \frac{30+2\sqrt{2\log(2/\delta)}}{\sqrt{n}}
    \end{equation*}
    with probability at least $1-\delta$.
\end{proposition}

\subsection{Platt scaling does not necessarily calibrate, even in \dce}

Platt scaling \citep{platt1999probabilistic} is a popular post-hoc calibration method which essentially consists of running a univariate logistic regression of observed $Y$ values on previous forecasts $f(X)$.\footnote{For finite samples, Platt scaling minimizes logistic loss with $\Tilde{Y}_i = Y_i\cdot \frac{\sum_{i=1}^n Y_i+1}{\sum_{i=1}^n Y_i+2} +(1-Y_i) \frac{1}{\sum_{i=1}^n (1-Y_i)+2}$ in place of $Y_i$ to improve empirical performance.}  
The issue, in short, is that there is no guarantee that Platt scaling will converge asymptotically to a calibrated function. Since $\dcem(f) = 0$ if and only $f$ is perfectly calibrated, this implies that Platt scaling cannot provide strong asymptotic calibration on \dce, \ece, and \calstar.
In \cref{fig:platt-scaling-hard}, we exhibit an example of a data-generation process for which the population Platt scaling algorithm yields an $h: [0,1]\rightarrow[0,1]$ (attained by optimizing logistic loss in expectation instead of over samples) that has non-zero dCE: $\dcem(h \circ f) > 0$. 

In \cref{apdx:post_hoc}, we propose a modified Platt scaling algorithm that asymptotically controls \calstar without data splitting, achieving the same theoretical guarantee as  isotonic regression for \calstar (\cref{prop:iso}).

\begin{figure}[t!]
    \centering
    \includegraphics[width=0.85\linewidth]{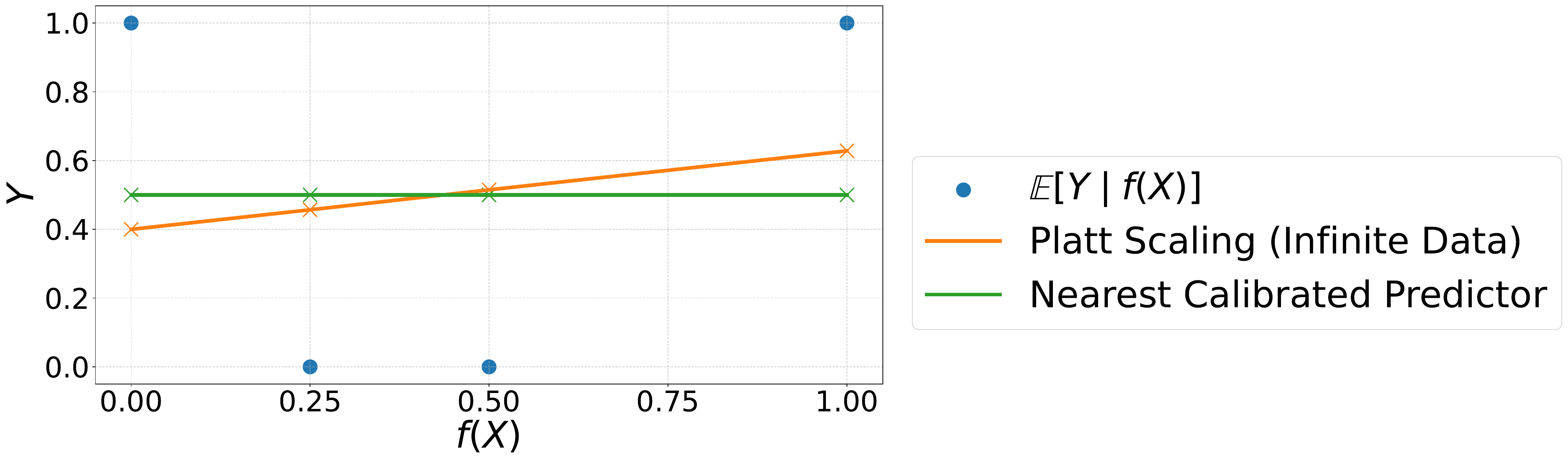}
    \caption{\textbf{Platt scaling does not guarantee small dCE even with infinite data.} In our example, we assume $f(X) \sim \mathrm{Uniform}\{0, 0.25, 0.5, 1\}$.}
    \label{fig:platt-scaling-hard}
\end{figure}

\section{Discussion}

There are several promising research directions and intriguing connections that stem from our findings. 
First, while our work focuses on the i.i.d.\ setting, we think our findings may hold ramifications for the branch of calibration research that focuses on sequential prediction and could potentially inform the design of new online calibration algorithms. 
Second, there may be connections between our work and \citet{sahoo2021reliable}, who also focus on ``threshold'' decision rules, although in a different setting than us, where forecasters provide a CDF estimate for a real-valued output.  
Third, the monotonic constraint used in our decision-theoretic guarantees may provide intriguing results in other contexts, such as that of \citet{roth2024forecasting}, who examine the impact of the dimensionality of a prediction space on swap regret and calibration. We also note that our focus on interval regions of the prediction space is related to the 1-dimensional special case in \citet{roth2024forecasting}, which builds on the work in \citet{noarov2023highdimensionalpredictionsequentialdecision} in showing that actionability is related to performance within convex subsets of the prediction space.
Finally, while the decision-theoretic guarantees of \calstar are broadest when $Y \in \{0,1 \}$, we believe that there may be a more general guarantee related to \calstar when $Y \in [0,1]$.

\section{Acknowledgments}
The authors gratefully acknowledge the National Science Foundation via grant DMS-2023109. J.S.\ and R.F.B.\ were partially supported by the Office of Naval Research via grant N00014-24-1-2544. J.S.\ was also partially supported by the Margot and Tom Pritzker Foundation. Z.R.\ was supported by NSF grant DMS-2413135. R.M.W.\ was partially supported by the NSF-Simons National Institute for Theory and Mathematics in Biology (NITMB) via grants NSF DMS-2235451 and Simons Foundation MP-TMPS-00005320.

\bibliographystyle{dcu}
\bibliography{refs}
\addcontentsline{toc}{section}{References}

\newpage
\appendix

\section{Relationships among calibration measures: proofs and examples\label{apdx:discrete_equivalence}}

\subsection{dCE and ECE are equivalent in the discrete setting}

In this section, we show that in the case where the prediction $f(X)$ is a discrete random variable, dCE and ECE are essentially equivalent.

\begin{proposition}
    Defining $b_f = \operatorname{ess}\inf\left\{|f(x)-f(x')| : x,x'\in\cX, f(x)\neq f(x')\right\}$,    \[\dcem(f) \geq \frac{b_f}{b_f+1} \cdot \ecem(f).\]\label{prop:dce_ece_equiv_discrete}
\end{proposition}

\begin{proof}
    Suppose $f(X)$ has finite support. We can write $f(X)\in\{u_1,\dots,u_m\}$ for some distinct values $u_1,\dots,u_m$, with $\min_{i\neq j}|u_i-u_j| \geq b_f$ by assumption.
    
    Fix any $\epsilon>0$. By definition of $\dcem(f)$, there exists some perfectly calibrated function $g$ such that
    \[\EE[|f(X)-g(X)|]\leq \dcem(f)+\epsilon.\]
    Define also a probability vector $p(g(X))$ with entries
    \[p_i(g(X)) = \PP\{f(X) = u_i \mid g(X)\},\]
    for $i\in[m]$. Moreover, define a random index $i^*(g(X)) \in\arg\max_{i\in[m]} p_i(g(X))$. That is, given $g(X)$, the index $i^*(g(X))$ identifies the (not necessarily unique) most likely value of $f(X)$ in the finite set $\{u_1,\dots,u_m\}$. Define also another random index, $i_{\textnormal{Med}}(g(X))$, satisfying that $u_{i_{\textnormal{Med}}(g(X))}$ is the (not necessarily unique) median of the distribution of $f(X)$ conditional on $g(X)$.
    
    Next, let $w$ be any function taking values in $[-1,1]$. We need to bound $\EE[w(f(X))\cdot(Y-f(X))]$. First, define
    \[\tilde{w}(g(X)) = w(u_{i^*(g(X))}).\]
    Next, we calculate
    \begin{align*}
        \PP\{w(f(X)) \neq \tilde{w}(g(X))\mid g(X)\}
        &\leq \PP\{f(X) \neq u_{i^*(g(X))} \mid g(X)\}\\
        &= 1 - p_{i^*(g(X))}(g(X))\\
        &= 1 - \max_{i\in[m]} p_i(g(X)).
    \end{align*}
    And,
    \begin{align*}
        \EE[|f(X)-g(X)| \mid g(X)]
        &\geq \inf_{t\in[0,1]} \EE[|f(X)-t| \mid g(X)]\\
        &=\EE[|f(X) - u_{i_{\textnormal{Med}}(g(X))}| \mid g(X)]\\
        &\geq \EE[b_f \cdot \ind{f(X) \neq u_{i_{\textnormal{Med}}(g(X))}} \mid g(X)]\\
        &= b_f \cdot \left(1 - p_{u_{i_{\textnormal{Med}}(g(X))}}(g(X)\right)\\
        &\geq b_f \cdot \left(1 - \max_{i\in[m]} p_i(X)\right),
    \end{align*}
    where the second step holds since the expected absolute loss is minimized by the median, and the third step holds by our assumption that $\min_{i\neq j}|u_i-u_j|\geq b_f$.
    Combining these calculations, we have shown that
    \[\PP\{w(f(X)) \neq \tilde{w}(g(X))\mid g(X)\} \leq \frac{1}{b_f} \cdot \EE[|f(X)-g(X)| \mid g(X)].\]
    After marginalizing over $g(X)$, then,
    \[\PP\{w(f(X)) \neq \tilde{w}(g(X))\} \leq \frac{1}{b_f} \cdot \EE{[|f(X)-g(X)|]} \leq \frac{1}{b_f}\left(\dcem(f) + \epsilon\right).\]
    Finally, we bound \ece. We have 
    \begin{multline*}
        \EE[w(f(X))\cdot(Y-f(X))]
        =\EE[\left(w(f(X)) - \tilde{w}(g(X))\right)\cdot(Y-f(X))] +{}\\ \EE[\tilde{w}(g(X))\cdot(g(X)-f(X))] + \EE[\tilde{w}(g(X))\cdot(Y-g(X))].\end{multline*}
    For the first term, since $\left|\left(w(f(X)) - \tilde{w}(g(X))\right)\cdot(Y-f(X))\right|\leq 1$ holds almost surely, we have
    \[\EE[\left(w(f(X)) - \tilde{w}(g(X))\right)\cdot(Y-f(X))] \leq \PP\{w(f(X)) \neq \tilde{w}(g(X))\} \leq \frac{1}{b_f}\left(\dcem(f) + \epsilon\right).\]
    For the second term, since $|\tilde{w}(g(X))|\leq 1$ almost surely, 
    \[\EE[\tilde{w}(g(X))\cdot(g(X)-f(X))] \leq \EE[|f(X)-g(X)|]\leq \dcem(f) + \epsilon.\]
    For the third term, since $g$ is perfectly calibrated,
    \[\EE[\tilde{w}(g(X))\cdot(Y-g(X))] = 0.\]
    Combining everything, then, we have shown that
    \[ \EE[w(f(X))\cdot(Y-f(X))] \leq \left(1 + \frac{1}{b_f}\right) \cdot \left(\dcem(f) + \epsilon\right).\]
    Since this holds for all $w$, and since $\epsilon>0$ can be taken to be arbitrarily small, this completes the proof.
\end{proof}

\subsection{Proof of \cref{prop:ineqs}: Relationship between \dce, \calstar and \ece\label{proof:ineqs}}

We break down the series of inequalities in \cref{prop:ineqs} into the following lemmas.

\begin{lemma}
    For any function $f$,
    \begin{equation*}
        \calstarm(f) \leq \ecem(f).
    \end{equation*}
\end{lemma}
\begin{proof}
    This follows immediately from \cref{sec:connections_wce}, which demonstrates that \calstar and ECE are both instances of wCE and that the complexity class of weight functions for ECE is a superset of the complexity class for \calstar.
\end{proof}

\begin{lemma}
For any function $f$,
    \begin{equation*}
        \dcem(f)^2/9 \leq \calstarm(f).
    \end{equation*}
\end{lemma}

\begin{proof}
Fix $N \in \mathbb N$.
Partition $[0,1]$ into $N$ equal-length intervals:
\begin{equation*}
    A_1 = [0,1/N]; A_2 = (1/N, 2/N];\dots; A_N = ((N-1)/N,1].
\end{equation*}
Define $\psi: [0,1] \rightarrow [N]$ as 
\begin{equation*}
    \psi(z) := \sum_{j=1}^N j \ind{z \in A_j},
\end{equation*}
i.e., $\psi(z)$ provides the index of which interval of $A_1,\dots, A_N$ that $z$ lies in.

Define $g: \mathcal{X} \rightarrow [0,1]$ as follows:
\begin{equation*}
    g(x) := \mathbb E [Y\mid \psi(f(X)) = \psi(f(x))].
\end{equation*}

Observe that $g$ is a calibrated:
\begin{equation*}
    \mathbb E[Y\mid g(X)] = \mathbb E[ \mathbb E [Y \mid \psi(f(X))] \mid g(X)] = \mathbb E[g(X)\mid g(X)] = g(X).
\end{equation*}

Define a random variable $V$ as 
\begin{equation*}
    V := \mathbb E[f(X) \mid \psi(f(X))].
\end{equation*}

We know that $\mathbb E |f(X) - V| \leq 1/N$ almost surely, since $\forall j \in [N]$ we have
\begin{equation*}
    \psi(f(X)) = j \implies f(X), V \in A_j \implies |f(X) - V| \leq \mathrm{diam}(A_j) = 1/N.
\end{equation*}

Therefore,
\begin{equation*}
    \mathbb E |g(X) - f(X)| \leq \mathbb E |g(X) - V| + \mathbb E |V - f(X)| \leq \mathbb E |g(X) - V| + 1/N.
\end{equation*}

Focusing on the first term,
\begin{align*}
    \mathbb E |g(X) - V| &= \mathbb E |\mathbb E [Y-f(X) \mid \psi(f(X))]|\\
    &= \sum_{j=1}^N \mathbb P \{f(X) \in A_j \} |\mathbb E [(Y-f(X))\mid f(X) \in A_j]|\\
    &= \sum_{j=1}^N |\mathbb E [(Y-f(X))\ind{f(X)\in A_j}]|\\
    &\leq N \calstarm(f).
\end{align*}
Therefore, since $g$ is calibrated,
\begin{equation*}
    \dcem(f) \leq N \calstarm(f) + 1/N.
\end{equation*}
Setting $N = \left\lceil \frac{1}{\sqrt{\calstarm(f)}} \right\rceil$, we get
\begin{equation*}
    \dcem(f) \leq 2 \sqrt{\calstarm(f)} + \calstarm(f) \leq 3 \sqrt{\calstarm(f)}.
\end{equation*}
\end{proof}

\subsection{Proof of \cref{prop:bv_fns}: Cutoff Calibration and bounded-variation functions \label{proof:bv_fns}}

\begin{proof}[\cref{prop:bv_fns}]
    The first inequality follows immediately from observing that indicator functions over intervals have a total variation of at most 2 and that, if $g \in \mathcal{B}_M$, then $-g \in \mathcal{B}_M$. 

    For the second inequality, by \cref{lemma:bv_funs_difference_representation}, $\exists g_1, g_2: [0,1] \rightarrow [0, \frac{M}{2}+1]$ non-decreasing such that $g = g_1 - g_2$. 

    Define $\mathcal{G} := \{g: [0,1] \rightarrow \mathbb [0, \frac{M}{2}+1] \mid g\text{ non-decreasing}   \}$.

    Define $\mathcal{G}' = \{g : [0, 1] \rightarrow [0, 1] \mid g \text{ non-decreasing}\}$.

    Therefore,
    \begin{align*}
        &\sup_{g \in \mathcal{B}_M} \mathbb E [(Y-f(X))g(f(X))] \\
        &\leq \sup_{g_1,g_2\in \mathcal{G}} \left(\mathbb E [(Y-f(X))g_1(f(X))] - \mathbb E [(Y-f(X))g_2(f(X))]\right) \\
        &= \sup_{g\in \mathcal{G}} \mathbb E [(Y-f(X))g(f(X))] + \sup_{g\in \mathcal{G}} \mathbb E [(Y-f(X))(-g(f(X)))] \\
        &\le 2\sup_{g\in \mathcal{G}\cup (-\mathcal{G})} \mathbb E [(Y-f(X))g(f(X))] \\
        &= 2 \sup_{g\in \mathcal{G}'\cup (-\mathcal{G}')} \mathbb E \left[(Y-f(X))(M/2+1)\cdot g(f(X))\right]\\
        &= (M+2) \sup_{g\in \mathcal{G}'\cup (-\mathcal{G}')} \mathbb E [(Y-f(X))g(f(X))]|\\
        &\leq (M+2) \calstarm(f).
    \end{align*}

The final inequality follows from observing both that
    \begin{equation*}
        \mathcal{G}' = \overline{\mathrm{conv}}\{g(x) = \ind{x \in [t,1]} \mid t \in [0,1] \}
    \end{equation*}
    and that we are taking the supremum of an objective function that is linear with respect to its argument $g$. More precisely,
\begin{align*}
    &\sup_{g\in \mathcal{G}'\cup (-\mathcal{G}')} \mathbb E [(Y-f(X))g(f(X))] \\
    &~~~= \max\left\{\sup_{g\in \mathcal{G}'} \mathbb E [(Y-f(X))g(f(X))], \sup_{g\in \mathcal{G}'} \mathbb E [(Y-f(X))(-g(f(X)))]\right\} \\
    &~~~= \max\left\{\sup_{t} \mathbb E [(Y-f(X))\ind{f(X) \in [t,1]}], \sup_{t} \mathbb E [(Y-f(X))(-\ind{f(X) \in [t,1]})]\right\} \\
    &~~~\le \sup_{t} |\mathbb E [(Y-f(X))\ind{f(X) \in [t,1]}]|. 
\end{align*}
\end{proof}

\begin{lemma}\label{lemma:bv_funs_difference_representation}
    Suppose $f: [0,1] \rightarrow [-1,1]$ and has total variation at most $M$. Then, $\exists f_1, f_2: [0,1] \rightarrow [0, \frac{M}{2}+1]$ non-decreasing such that $f = f_1 - f_2$.
    
\end{lemma}
\begin{proof}
    We proceed similarly to \citet[Section 32 Theorem 4]{kolmogorov1975introductory}.

    Let $v(x)$ be the total variation of $f$ on $[0,x]$. Observe that $\mathrm{im}(v) \subseteq [0, M]$.

    Define $f_1 = 0.5(v+f)+0.5$ and $f_2 = 0.5(v-f)+0.5$.

    We know $f_2$ is non-decreasing from \citet[Section 32 Theorem 4]{kolmogorov1975introductory}.

    To show $f_1$ is non-decreasing, let $x \leq y$, where $x,y \in [0,1]$.

    Then,
    \begin{equation*}
        (v+f)(y) - (v+f)(x) = v(y) - v(x) + f(y) - f(x)\geq 0,
    \end{equation*}
    since $|f(y) - f(x)| \leq v(y) - v(x)$ \citep[Section 32 Theorem 4]{kolmogorov1975introductory}.

    Conclude by noting $f = f_1 - f_2$.
\end{proof}

\subsection{Examples that demonstrate tightness\label{apdx:examples_tightness}}

\paragraph{Example where dCE, \calstar, and ECE are equal} Suppose $Y \equiv 0$ and $f(x) \equiv 1$. Then, the only calibrated $g: \mathcal{X} \rightarrow [0,1]$ is $g(x)\equiv 0$. Therefore, $\dcem(f) = 1$. Similarly, $|\mathbb E [(Y-f(X))\ind{f(X) \in I}]|$ for some interval $I$ is upper bounded by $|\mathbb E[(Y-f(X))]| = 1$; therefore, $\calstarm(f) = 1$. Finally, $\ecem(f) = \mathbb E |1 - 0| = 1$.
Therefore,
\[
\dcem(f) = \calstarm(f) = \ecem(f).
\]

\paragraph{Example showing \cref{prop:dce_ece_equiv_discrete} is tight} Let $\cX = [0,1]$, with $X\sim\textnormal{Unif}[0,1]$ and $Y = \ind{X>0.5}$. Let $f(x) = 0.5(1+b) - b\cdot\ind{X>0.5}$, which satisfies the assumption (i.e., the separation is $\geq b$). Let $g(x) \equiv 0.5$, which is perfectly calibrated. Clearly $\dcem(f)$ is attained by comparing to $g$, i.e., $\dcem(f) = \EE|f(X)-g(X)| = 0.5b$. And, $\ecem(f) = 0.5(1+b)$, so,
    \[\ecem(f) = \left(1+\frac{1}{b}\right)\dcem(f).\]

\paragraph{Example showing the quadratic relationship between dCE and \calstar} Suppose $f(X) \sim \mathrm{Uniform}(0,1)$. Suppose $\mathbb E [Y\mid f(X)] = \frac{i-0.5}{N}$ for $f(X) \in (\frac{i-1}{N}, \frac{i}{N}]$. Then, $\calstarm(f) = |\mathbb E[(Y-f(X))\ind{f(X) \in (0, \frac{0.5}{N}}]| = \frac{1}{8N^2}$. Meanwhile, $\dcem(f) \asymp \frac{1}{N}$. 

\section{Decision-theoretic guarantees: proofs and extensions \label{apdx:actionable_extensions}}

\subsection{Proof of \cref{prop:decision_theory_ece}: ECE upper bounds the risk gap for binary decision loss\label{proof:decision_theory}}

\begin{proof}[\cref{lemma:decision_theory_rearrangement}]
First observe that
\begin{align*}
    &R_\mathrm{bd}(f; \tau) - R_\mathrm{bd}(h \circ f; \tau)\\
    &= \tau \EE [(1-Y)(\ind{f(X) \geq \tau} - \ind{h\circ f(X) \geq \tau})] + (1-\tau) \EE [Y(\ind{f(X) < \tau} - \ind{h\circ f(X) < \tau})].
\end{align*}

Then, turning differences of indicators into products of indicators, we get
\begin{equation*}
    \begin{split}
        R_\mathrm{bd}(f; \tau) - R_\mathrm{bd}(h \circ f; \tau) = &\tau \mathbb E[(1-Y)\ind{f(X) \geq \tau}\ind{h\circ f(X) < \tau}]\\
        &-\tau\EE[(1-Y)\ind{f(X)<\tau}\ind{h\circ f(X)\geq \tau}]\\
        &+ (1-\tau) \mathbb E[Y \ind{f(X) < \tau}\ind{h\circ f(X)\geq \tau}]\\
        &-(1-\tau) \mathbb E[Y\ind{f(X)\geq\tau}\ind{h \circ f(X) < \tau}].
    \end{split}
\end{equation*}
    We achieve the desired result by distributing $(1-\tau)$ and $(1-Y)$ and collecting like terms.
\end{proof}

\begin{proof}[\cref{prop:decision_theory_ece}]
For any function $h$,
    \begin{align*}
        &R_\mathrm{bd}(f; \tau) - R_\mathrm{bd}(h \circ f; \tau)\\
        &= \mathbb E [(\mathbb E[Y\mid f(X)] - \tau)\ind{f(X)<\tau}\ind {h \circ f(X)\geq \tau}] + \mathbb E [(\tau - \mathbb E[Y\mid f(X)])\ind{f(X)\geq\tau}\ind{h \circ f(X)<\tau}]\\
        &\leq \mathbb E [(\mathbb E[Y\mid f(X)] - \tau)\ind{f(X)<\tau}\ind{\mathbb E[Y\mid f(X)]\geq \tau}] + \mathbb E [(\tau - \mathbb E[Y\mid f(X)])\ind{f(X)\geq\tau}\ind{\mathbb E[Y\mid f(X)]<\tau}]\\
        &\leq \mathbb E [(\mathbb E[Y\mid f(X)] - f(X))\ind{f(X)<\tau}\ind{\mathbb E[Y\mid f(X)]\geq \tau}] + \mathbb E [(f(X) - \mathbb E[Y\mid f(X)])\ind{f(X)\geq\tau}\ind{\mathbb E[Y\mid f(X)]<\tau}]\\
        &\leq \mathbb E [(\mathbb E[Y\mid f(X)] - f(X))\ind{f(X)<\mathbb E[Y\mid f(X)]}] + \mathbb E [(f(X) - \mathbb E[Y\mid f(X)])\ind{f(X)> \mathbb E[Y\mid f(X)]}]\\
        &= \ecem(f).
    \end{align*}
\end{proof}

\subsection{Proof of \cref{prop:decision_theory_calstar}: \calstar upper bound the monotone risk gap for binary decision loss}

\begin{proof}[\cref{prop:decision_theory_calstar}]
Let $h: [0,1]\rightarrow [0,1]$ be monotone.

    Define $I,I' \subseteq [0,1]$ such that $\ind{h\circ f(X)\geq \tau} = \ind{f(X) \in I}$ and $\ind{h\circ f(X) < \tau} = \ind{f(X) \in I'}$. Since $h$ is monotone, we know that $I,I'$ must be intervals.

    Then, using \cref{lemma:decision_theory_rearrangement},
    \begin{align*}
        &R_\mathrm{bd}(f; \tau) - R_\mathrm{bd}(h\circ f; \tau)\\
        &= \mathbb E [(Y - \tau)\ind{f(X)<\tau}\ind{h\circ f(X)\geq \tau}] + \mathbb E [(\tau - Y)\ind{f(X)\geq\tau}\ind{h\circ f(X)<\tau}]\\
        &\leq \mathbb E [(Y - f(X))\ind{f(X)<\tau}\ind{h\circ f(X)\geq \tau}] + \mathbb E [(f(X) - Y)\ind{f(X)\geq\tau}\ind{h(f)<\tau}]\\
        &= \mathbb E [(Y - f(X))\ind{f(X)<\tau}\ind{f(X) \in I}] + \mathbb E [(f(X) - Y)\ind{f(X)\geq\tau}\ind{f(X) \in I'}]\\
        &\leq 2\calstarm(f).
    \end{align*}
\end{proof}

\subsection{\calstar upper bounds the monotone risk gap for bounded proper scoring rules, when $Y \in \{0,1 \}$}

We now generalize our results for $\ell_\mathrm{bd}$ to a broader class of loss functions, still in the context of $Y \in \{ 0,1\}$. Like \citet{hu2024calibrationerrordecisionmaking} and \citet{kleinberg2023u}, we consider a loss function $\ell: \{0,1 \} \times [0,1] \rightarrow \mathbb R$ that follows the very minimal constraint that it be a proper scoring rule:
\begin{equation*}
    \mathbb E_{Y \sim \mathrm{Bernoulli}(p)}[\ell(Y, p)] \leq  \mathbb E_{Y \sim \mathrm{Bernoulli}(q)}[\ell(Y, p)].
\end{equation*}

While it may seem hopeless to be able to find a rearrangement like \cref{lemma:decision_theory_rearrangement} that will aid us in relating the risk gap to a calibration measure for this minimally constrained $\ell$, the fact that $\ell$ is a proper scoring rule proves to be enough. Loss functions that are proper scoring rules can be represented as mixtures across $\tau$ of $\ell_\mathrm{bd}(\cdot, \cdot; \tau)$: binary decision loss forms a basis for proper scoring rules. This representation dates back to \citet{schervish1989general}. Using a formulation of the Schervish representation from \citet{gneiting2007strictly}, we get the following result.

\begin{proposition}
    Suppose $Y \in \{0,1\}$. Suppose $\ell: \{0,1 \} \times [0,1] \rightarrow \mathbb R$ is a proper scoring rule that also satisfies the regularity conditions listed in \citet[Theorem 3]{gneiting2007strictly}. If the measure $\nu$ used to define the Schervish representation is a probability measure, then
    \begin{equation*}
        \mathbb E[\ell(Y, f(X))] - \inf_{h} \mathbb E[\ell(Y, h \circ f(X))] \leq \ecem(f)
    \end{equation*}
     and
     \begin{equation*}
        \mathbb E[\ell(Y, f(X))] - \inf_{h: \textnormal{ monotone}} \mathbb E[\ell(Y, h \circ f(X))] \leq 2\calstarm(f).
    \end{equation*}
\end{proposition}
\begin{proof}
    Using the Schervish representation, we have
    \begin{equation*}
        \mathbb E[\ell(Y, f(X))] - \mathbb E[\ell(Y, h \circ f(X))] = \mathbb E\left[ \mathbb E_{\tau \sim \nu}[\ell_{\mathrm{bd}}(Y, \ind{f(X)\geq \tau}, \tau) - \ell_{\mathrm{bd}}(Y, \ind{h \circ f(X)\geq \tau}, \tau)]\right] 
    \end{equation*}
    By Fubini's theorem,
    \begin{align*}
        &\mathbb E\left[ \mathbb E_{\tau \sim \nu}[\ell_{\mathrm{bd}}(Y, \ind{f(X)\geq \tau}, \tau) - \ell_{\mathrm{bd}}(Y, \ind{h \circ f(X)\geq \tau}, \tau)]\right]\\
        &= \mathbb E_{\tau \sim \nu}\left[ \mathbb E[\ell_{\mathrm{bd}}(Y, \ind{f(X)\geq \tau}, \tau) - \ell_{\mathrm{bd}}(Y, \ind{h \circ f(X)\geq \tau}, \tau)]\right]\\
        &\leq \sup_{\tau\in [0,1]} \mathbb E[\ell_{\mathrm{bd}}(Y, \ind{f(X)\geq \tau}, \tau) - \ell_{\mathrm{bd}}(Y, \ind{h \circ f(X)\geq \tau}, \tau)]
    \end{align*}
    We conclude by applying \cref{prop:decision_theory_ece} and \cref{prop:decision_theory_calstar}.
\end{proof}

\subsection{\calstar upper bounds the monotone risk gap of a sign-testing loss}

When $Y \in [0,1]$, one may operate under the following sign-testing risk function (the ``st'' subscript refers to ``sign testing''), with two actions and a pre-specified threshold $Y^*$,
\begin{equation*}
    R_\mathrm{st}(f; Y^*) = \mathbb E [(Y-Y^*)\ind{f(X) \leq Y^*} \ind{Y>Y^*}] + \mathbb E [(Y^* - Y)\ind{f(X)>Y^*} \ind{Y \leq Y^*}].
\end{equation*}

Intuitively, we are trying to predict the sign of $Y-Y^*$, and the price we pay when we guess incorrectly is $|Y-Y^*|$.

Then, after rearrangement, we once again get, for arbitrary $h: [0,1]\rightarrow[0,1]$, 
\begin{equation*}
    \begin{split}
        R_\mathrm{st}(f; Y^*) - R_\mathrm{st}(h\circ f; Y^*) = &\mathbb E [(Y-Y^*)\ind{f(X) \leq Y^*} \ind{h \circ f(X)>Y^*}]\\ 
        &+ \mathbb E [(Y^* - Y)\ind{f(X)>Y^*}\ind{h \circ f(X)\leq Y^*}].
    \end{split}
\end{equation*}

We can use this identity to recover the same theoretical guarantees, in terms of \calstar and \ece, as in the original binary decision framework.

\begin{proposition} For any $Y^* \in [0,1]$,
    \begin{equation*}
        R_\mathrm{st}(f; Y^*) - \inf_{h \textnormal{ monotone}}R_\mathrm{st}(h \circ f; Y^*) \leq 2 \calstarm(f)
    \end{equation*}
    and
    \begin{equation*}
        R_\mathrm{st}(f; Y^*) - \inf_h R_\mathrm{st}(h\circ f; Y^*) \leq \ecem(f).
    \end{equation*}
\end{proposition}

The proof follows using the exact steps as those used to prove \cref{prop:decision_theory_ece} and \cref{prop:decision_theory_calstar}.

The fact that we recover the same theoretical guarantees as before is somewhat intuitive, given that the Bayes decision rule here has the same form as before: $\ind{\mathbb E[Y\mid f] \geq Y^*}$.

\subsection{Implementation details of the experiment\label{apdx:experiment_details}}
In this section, we give details for the implementation of the experiment in \cref{sec:experiment}.

 In addition to the samples we used to fit $f$, we produce $N=1000$ independent samples of $(X,Y)$ pairs. 

In all of the following procedures to approximate population quantities, we use $\mathbb E[Y\mid X]$ in place of $\mathbb E[Y\mid f(X)]$ since $f$ is injective and $X$ is a continuous random variable. Note that our procedures are sample efficient since we are plugging in $\mathbb E[Y\mid f(X)]$ in place of $Y$.

We approximate $\dcem(f)$ using
\begin{equation*}
    \sup_{w \in \mathcal{L}_1} \frac{1}{N} \sum_{j=1}^N w(f(X_j))(\mathbb E[Y_j\mid X_j] - f(X_j)).
\end{equation*}

We approximate $\calstarm(f)$ using
\begin{equation*}
    \sup_{I: \textnormal{ interval}} \frac{1}{N} \sum_{j=1}^N (\mathbb E[Y_j\mid X_j] - f(X_j))\ind{f(X_j) \in I}.
\end{equation*}

We approximate $\ecem(f)$ using
\begin{equation*}
        \frac{1}{N} \sum_{j=1}^N |\mathbb E[Y_j\mid X_j] - f(X_j)|.
\end{equation*}

We approximate $R_\mathrm{bd}(f; \tau)$ using
\begin{equation*}
    \frac{1}{N}\sum_{j=1}^N \ell_\mathrm{bd}(\mathbb E[Y_j\mid X_j], \ind{f(X_j)\geq \tau}; \tau).
\end{equation*}

We approximate $\inf_h R_\mathrm{bd}(h \circ f; \tau)$ using
\begin{equation*}
    \frac{1}{N}\sum_{j=1}^N \ell_\mathrm{bd}(\mathbb E[Y_j\mid X_j], \ind{\mathbb E[Y_j \mid X_j]\geq \tau}; \tau).
\end{equation*}

To approximate $\inf_{h: \textnormal{ monotone}} R_\mathrm{bd}(h \circ f; \tau)$, we minimize over all monotone decision rules:
\begin{equation*}
        \begin{split}
         \min\biggl(&\min_{\tau' \in \{0, f(X_1),\dots, f(X_N), 1\}}\frac{1}{N}\sum_{j=1}^N \ell_\mathrm{bd}(\mathbb E[Y_j\mid X_j], \ind{f(X_j)\geq \tau'}; \tau),\\ &\min_{\tau' \in \{0, f(X_1),\dots, f(X_N), 1\}}\frac{1}{N}\sum_{j=1}^N \ell_\mathrm{bd}(\mathbb E[Y_j\mid X_j], \ind{f(X_j)\leq \tau'}; \tau) \biggr).
        \end{split}
\end{equation*}

\section{Testability: proofs\label{apdx:testability}}

\subsection{Proof of \cref{thm:hardness_ece}: constraints on the effective support size of algorithms that strongly asymptotically control ECE \label{proof:support_size_ece}}

First we need a lemma, which we will prove in the end. Here, and for the remainder of the proof of this theorem, we will always write $\ecem^P$ rather than $\ecem$ to clarify which distribution $P$ is being used for computing the \ece. 
\begin{lemma}\label{lem:use_sample_resample}
Fix any $N> n\geq 1$. Let     \begin{equation*}
        \mathcal{A}:(\mathcal{X} \times[0,1])^n \rightarrow\{\text { measurable functions } f: \mathcal{X} \rightarrow[0,1]\}
    \end{equation*}
    be any procedure that inputs a data set $(X_1,Y_1),\dots,(X_n,Y_n)$, and returns a fitted function $f_n = \mathcal{A}\big((X_1,Y_1),\dots,(X_n,Y_n)\big)$. 
Let $(X_1,Y_1),\dots,(X_N,Y_N)\overset{iid}{\sim} P$, let $f_n = \mathcal{A}\big((X_1,Y_1),\dots,(X_n,Y_n)\big)$ be trained on the first $n$ data points, and let $\widehat{P}_{N-n} = \frac{1}{N-n}\sum_{i=n+1}^N \delta_{(X_i,Y_i)}$ be the empirical distribution of the remaining $N-n$ data points. then
\[
\EE_{P^N}[\ecem^{\widehat{P}_{N-n}}(f_n)]  \leq \frac{\sup_Q \EE_{Q^n}[\ecem^Q(f_n)]}{1-\frac{n(n-1)}{2N}} + \frac{2n}{N}.\]
\end{lemma}
Under the assumptions of \cref{thm:hardness_ece}, since $\ecem\leq 1$ always, we must have
\[
\sup_Q \EE_{Q^n}[\ecem^Q(f_n)] \leq c + \sup_Q \PP_{Q^n}(\ecem^Q(f_n)>c) \leq c+ \delta,\]
and consequently, 
\[\EE_{P^N}[\ecem^{\widehat{P}_{N-n}}(f_n)] \leq \frac{c+\delta}{1-\frac{n(n-1)}{2N}} + \frac{2n}{N}.\]

Our next step is to work with this \ece error for the empirical distribution $\widehat{P}_{N-n}$.
First consider a fixed function $f$. If for some $i\in\{n+1,\dots,N\}$, $f(X_i)$ is unique among the values $f(X_{n+1}),\dots,f(X_N)$, then $\EE_{\widehat{P}_{N-n}}[Y \mid f(X) = f(X_i)] = Y_i$. In particular, this implies
\[\ecem^{\widehat{P}_{N-n}}(f) \geq \frac{1}{N-n}\sum_{i=n+1}^N |Y_i - f(X_i)| \cdot \ind{f(X_i) \neq f(X_j),~\forall j\in\{n+1,\dots,N\}\backslash \{i\}}.\]
Therefore, for a fixed function $f$,
\[\EE\left[\ecem^{\widehat{P}_{N-n}}(f)\right] = \frac{1}{N-n}\sum_{i=n+1}^N \EE\left[|Y_i - f(X_i)| \cdot \ind{f(X_i) \neq f(X_j),~ \forall j\in\{n+1,\dots,N\}\backslash \{i\}}\right],\]
or equivalently by symmetry,
\[\EE\left[\ecem^{\widehat{P}_{N-n}}(f)\right] = \EE\left[|Y_{n+1} - f(X_{n+1})| \cdot \ind{f(X_{n+1}) \neq f(X_j),~ \forall j\in\{n+2,\dots,N\}}\right].\]
Now define $p(x) = \PP_P(f(X)=f(x))$. Then
\[\EE\left[\ind{f(X_{n+1}) \neq f(X_j),~ \forall j\in\{n+2,\dots,N\}} \mid X_{n+1}\right]
= (1-p(X_{n+1}))^{N-n-1} \geq e^{-1} \cdot \ind{p(X_{n+1}) \leq \frac{1}{N-n}},\]
since $(1-\frac{1}{m})^{m-1}\geq e^{-1}$ for all $m\geq 1$, and therefore,
\[\EE\left[\ecem^{\widehat{P}_{N-n}}(f)\right] \geq e^{-1} \cdot \EE\left[|Y_{n+1} - f(X_{n+1})| \cdot \ind{p(X_{n+1}) \leq \frac{1}{N-n}}\right].\]
Since $Y\in[0,1]$, for any $x$, we must have
\begin{align*}
\textnormal{Var}(Y\mid X=x) 
&= \EE[|Y-\EE[Y\mid X]|^2\mid X=x] \\
&= \inf_{t\in[0,1]} \EE[|Y-t|^2\mid X=x] \leq  \inf_{t\in[0,1]} \EE[|Y-t|\mid X=x], 
\end{align*}
and so
\[\EE\left[\ecem^{\widehat{P}_{N-n}}(f)\right] \geq e^{-1} \cdot \EE\left[\textnormal{Var}(Y_{n+1} \mid X_{n+1}) \cdot \ind{p(X_{n+1}) \leq \frac{1}{N-n}}\right].\]

Note that we must have $< N-n$ many values $t\in[0,1]$ such that $\PP_P(f(X)=t)> \frac{1}{N-n}$,
so 
\[\PP_P\Big(p(X) \leq \frac{1}{N-n}\Big) \geq 1 - \sup_{t_1,\dots,t_{N-n-1}}\PP_P(f(X)\in\{t_1,\dots,t_{N-n-1}\}).\]
If $S_\gamma(f,P) \geq  N-n$, then, we must have $\sup_{t_1,\dots,t_{N-n-1}}\PP_P(f(X)\in\{t_1,\dots,t_{N-n-1}\}) < 1-\gamma$ and so $\PP_P(p(X) \leq \frac{1}{N-n}) \geq \gamma$.
Therefore, 
\[\textnormal{If $S_\gamma(f,P) \geq  N-n$, then $\EE\left[\textnormal{Var}(Y_{n+1}|X_{n+1}) \cdot \ind{p(X_{n+1}) \leq \frac{1}{N-n}}\right] \geq \sigma^2_\gamma(P)$},\]
by definition of $\sigma^2_\gamma(P)$. In other words,
\[\textnormal{If $S_\gamma(f,P) \geq  N-n$, then $\EE\left[\ecem^{\widehat{P}_{N-n}}(f)\right] \geq e^{-1}\sigma^2_\gamma(P)$.}\]
Now we plug in $f=f_n$, and condition on $f_n$:
\[\textnormal{If $S_\gamma(f_n,P) \geq  N-n$, then $\EE\left[\ecem^{\widehat{P}_{N-n}}(f_n)\mid f_n\right] \geq e^{-1}\sigma^2_\gamma(P)$.}\]
(Here, since we are conditioning on $f_n$, the expected value is being taken with respect to the remaining random sample, $(X_{n+1},Y_{n+1}),\dots,(X_N,Y_N)$.) Finally, marginalizing over $f_n$,
\[\EE\left[\ecem^{\widehat{P}_{N-n}}(f_n)\right] \geq e^{-1}\sigma^2_\gamma(P)\cdot\PP(S_\gamma(f_n,P) \geq  N-n).\]

Returning to the results of the lemma, then,
\[e^{-1}\sigma^2_\gamma(P)\cdot\PP(S_\gamma(f_n,P) \geq  N-n) \leq \frac{c+\delta}{1-\frac{n(n-1)}{2N}} + \frac{2n}{N}.\]
Finally, choosing $N=n^2$, we obtain
\[e^{-1}\sigma^2_\gamma(P)\cdot\PP(S_\gamma(f_n,P) \geq  n^2) \leq \frac{c+\delta}{1-\frac{n(n-1)}{2n^2}} + \frac{2}{n} \leq 2(c+\delta + n^{-1}),\]
which completes the proof.

Finally, we prove the lemma. 

\begin{proof}[Lemma~\ref{lem:use_sample_resample}]
    First, let $\widehat{P}_N = \frac{1}{N}\sum_{i=1}^N \delta_{(X_i,Y_i)}$ be the empirical distribution of the full sample of size $N$. For the moment, treat $(X_1,Y_1),\dots,(X_N,Y_N)$ as fixed. Then
    \[\EE_{(\widehat{P}_N)^n}[\ecem^{\widehat{P}_N}(f_n)] \leq \sup_Q \EE_{Q^n}[\ecem^Q(f_n)]\]
    since $Q=\widehat{P}_N$ is a valid distribution. To be more explicit, on the left-hand side, we have $f_n$ being the output of $\mathcal{A}$ when trained on a sample of size $n$ drawn i.i.d.\ from $\widehat{P}_N$, or equivalently,
    \[f_n = \mathcal{A}((X_{i_1},Y_{i_1}),\dots,(X_{i_n},Y_{i_n}))\]
    where $i_1,\dots,i_n\overset{iid}{\sim} \mathrm{Uniform}([N])$. Now, 
    with probability $(1-\frac{1}{N})\cdot\hdots\cdot(1-\frac{n-1}{N}) \geq 1 - \frac{n(n-1)}{2N}$, the indices $i_1,\dots,i_n$ are all distinct---i.e., equivalently, these indices have been sampled
    without replacement (WOR). Therefore, 
    \begin{multline*}\EE_{i_1,\dots,i_n\sim\textnormal{WOR}}[\ecem^{\widehat{P}_N}(\mathcal{A}((X_{i_1},Y_{i_1}),\dots,(X_{i_n},Y_{i_n})))] \\
    = \EE_{(\widehat{P}_N)^n}[\ecem^{\widehat{P}_N}(\mathcal{A}((X_{i_1},Y_{i_1}),\dots,(X_{i_n},Y_{i_n})))\mid i_1,\dots,i_n\textnormal{ distinct}]
    \\\leq \frac{\EE_{(\widehat{P}_N)^n}[\ecem^{\widehat{P}_N}(\mathcal{A}((X_{i_1},Y_{i_1}),\dots,(X_{i_n},Y_{i_n})))\cdot\ind{i_1,\dots,i_n\textnormal{ distinct}}]}{1 - \frac{n(n-1)}{2N}} \\\leq \frac{\EE_{(\widehat{P}_N)^n}[\ecem^{\widehat{P}_N}(\mathcal{A}((X_{i_1},Y_{i_1}),\dots,(X_{i_n},Y_{i_n})))]}{1 - \frac{n(n-1)}{2N}},\end{multline*}
    where the first expected value is taken with respect to $i_1,\dots,i_n$ sampled uniformly without replacement from $[N]$. 
    
    Now, take the expected value over $(X_1,Y_1),\dots,(X_N,Y_N)\overset{iid}{\sim}P$. We then have
    \[\EE\left[\ecem^{\widehat{P}_N}(\mathcal{A}((X_{i_1},Y_{i_1}),\dots,(X_{i_n},Y_{i_n})))\right] \leq \frac{\sup_Q \EE_{Q^n}[\ecem^Q(f_n)]}{1-\frac{n(n-1)}{2N}}\]
    where, on the left-hand side, the expected value is taken with respect to both the draw of 
    $(X_1,Y_1),\dots,$ $(X_N,Y_N)\overset{iid}{\sim}P$, and also $i_1,\dots,i_n$ sampled uniformly without replacement from $[N]$. By symmetry, the data points are invariant to permutation, so it's equivalent to write
    \[\EE_{P^N}\left[\ecem^{\widehat{P}_N}(f_n)\right] \leq \frac{\sup_Q \EE_{Q^n}[\ecem^Q(f_n)]}{1-\frac{n(n-1)}{2N}}\]
    where now on the left-hand side, $f_n=\mathcal{A}((X_1,Y_1),\dots,(X_n,Y_n))$ is simply trained
    on the first $n$ data points.

    Finally, let $P$ and $Q$ be any distributions, and let $f:\mathcal{X}\to[0,1]$ be any function. Then for any function $w:[0,1]\to[-1,1]$,
    \[\EE_P[w(f(X))\cdot(Y-f(X))] \leq \EE_Q[w(f(X))\cdot(Y-f(X))] + 2d_{TV}(P,Q),\]
    since $w(f(X))\cdot(Y-f(X))\in[-1,1]$. So,
    \begin{align*}
        \ecem^P(f)& = \sup_{w:[0,1]\to[-1,1]}\EE_P[w(f(X))\cdot(Y-f(X))]\\
        &\leq \sup_{w:[0,1]\to[-1,1]}\EE_Q[w(f(X))\cdot(Y-f(X))] + 2d_{TV}(P,Q)\\
        &=\ecem^Q(f) + 2d_{TV}(P,Q).
    \end{align*}
    Consequently,  it holds almost surely that
    \[\ecem^{\widehat{P}_{N-n}}(f_n) \leq \ecem^{\widehat{P}_N}(f_n) + 2d_{TV}(\widehat{P}_{N-n},\widehat{P}_N) \leq \ecem^{\widehat{P}_N}(f_n) + \frac{2n}{N}.\]
    This completes the proof.
\end{proof}

\subsection{Hardness of asymptotic \ece calibration}
Our hardness result for \ece, stated in \cref{thm:hardness_ece}, is closely
related to an asymptotic hardness result of \citet{gupta2020distribution}. 
Here we compare between the two, and show how our result implies the existing result---in particular, 
ours gives a more precise characterization by providing a finite-sample bound on effective
support size of any procedure that guarantees \ece calibration, but asymptotically can be
interpreted in an analogous way.

First we define what it means for an algorithm to provide an asymptotic \ece guarantee.
\begin{definition}
    Consider an algorithm 
    \begin{equation*}
        \mathcal{A}: \cup_{n \geq 0}(\mathcal{X} \times[0,1])^n \rightarrow\{\text { measurable functions } f: \mathcal{X} \rightarrow[0,1]\}.
    \end{equation*}
    We say that \textbf{an algorithm $\mathcal{A}$ is strongly asymptotically \ece-calibrated} if
    \begin{equation*}
        \limsup_{n\rightarrow\infty} \sup_{P} \mathbb E_{P^n} [\ecem(f_n)]= 0,
    \end{equation*}
    where $\sup_P$ takes the supremum over all distributions $P$ on $\mathcal{X} \times [0,1]$ and where $f_n := \mathcal{A}(\mathcal D_n)$ for data set $\mathcal{D}_n :=  \{(X_i, Y_i) \}_{i \in [n]} \sim P^n$.\label{def:strong_asymptotic_calibration}
\end{definition}

\citet{gupta2020distribution} show that any post-hoc calibration method returning an injective
map on the pretrained model, cannot satisfy this property---that is, if $\mathcal{A}$ has
the potential to produce a map whose output is continuously distributed on $[0,1]$.
We will now see how \cref{thm:hardness_ece} implies this same takeaway message: the 
asymptotic calibration property is impossible to attain unless we bound the rate at which the support size of the output $f_n$ can grow.
\begin{corollary}\label{cor:hardness_ece}
    Suppose $\mathcal{A}$ is strongly asymptotically \ece-calibrated (in the sense of \cref{def:strong_asymptotic_calibration}).  
    Then there exists a sequence $\nu_n\to 0$ such that
    \[\PP_{P^n}\left\{S_\gamma(f_n,P) \leq n^2 \right\}\geq 1- \frac{\nu_n}{\sigma^2_\gamma(P)}\textnormal{ for all distributions $P$ and all $\gamma$}.\]
\end{corollary}
We can interpret this as saying that any function $f_n$, returned by a post-hoc calibration procedure with asymptotic \ece guarantees, must necessarily
be (mostly) discrete.
\begin{proof}
    Define
    \[\epsilon_n =  \sup_P \EE_{P^n}[\ecem^P(f_n)].\]
    If $\mathcal{A}$ is strongly asymptotically calibrated, then we must have
    \[\epsilon_n\to 0.\]
    Then
    \[\sup_P \EE_{P^n}[\ecem^P(f_n)]\leq \epsilon_n \ \Longrightarrow \  \PP_{P^n}(\ecem^P(f_n)\leq \epsilon_n^{1/2}) \geq 1-\epsilon_n^{1/2}\textnormal{ for all $P$},\]
    by Markov's inequality. Applying \cref{thm:hardness_ece} with $c=\delta=\epsilon_n^{1/2}$, then, for all $\gamma$,
    \[\PP_{P^n}\left\{S_\gamma(f_n,P) \leq n^2 \right\} \geq 1 - \frac{2e(2\epsilon_n^{1/2}+n^{-1})}{\sigma^2_\gamma(P)}\textnormal{ for all distributions $P$}.\]
    Since $2e(2\epsilon_n^{1/2}+n^{-1})\to 0$ this completes the proof.
\end{proof}

\subsection{Proof of \cref{cor:guarantee_calstar}}
We will prove the stronger statement
\[\PP\big\{\calstarm(f_n)\leq c \mid (X_1,Y_1),\dots,(X_{\lceil n/2\rceil},Y_{\lceil n/2\rceil})\big\} \geq 1-2\delta.\]
In other words, we will condition on the first batch of data (used to train $f_{n,\textnormal{init}}$), and the probability is computed with respect to the distribution
 the remaining half of the data, $(X_{\lceil n/2\rceil+1},Y_{\lceil n/2\rceil+1}), \dots,$ $(X_n,Y_n) \overset{iid}{\sim} P $.

We will define two events:
\[\mathcal{E}_1 = \left\{\left|\calstarhat(f_{n,\textnormal{init}}) - \calstarm^P(f_{n,\textnormal{init}}) \right| 
        \leq \frac{20+\sqrt{2\log(1/\delta)}}{\sqrt{\lfloor n/2\rfloor}}\right\}\]
        and
\[\mathcal{E}_2 = \left\{ \left| \frac{1}{\lfloor n/2\rfloor}\sum_{i=\lceil n/2\rceil +1}^n  Y_i - \EE_P[Y]\right|\leq \sqrt{\frac{\log(1/\delta)}{2\lfloor n/2\rfloor}}\right\}.\]
By Proposition~\ref{prop:estimation} (applied to the data set $\{(X_i,Y_i)\}_{i=\lceil n/2\rceil + 1}^n$, i.e., we have sample size $\lfloor n/2\rfloor$ in place of $n$), we have
\[\PP\big\{\mathcal{E}_1 \mid (X_1,Y_1),\dots,(X_{\lceil n/2\rceil},Y_{\lceil n/2\rceil})\big\} \geq 1-\delta.\]
Moreover, by Hoeffding's inequality (again, applied with sample size $\lfloor n/2\rfloor$), we have
\[\PP\big\{\mathcal{E}_2 \mid (X_1,Y_1),\dots,(X_{\lceil n/2\rceil},Y_{\lceil n/2\rceil})\big\} \geq 1-\delta.\]
To complete the proof, then, it suffices to show that on the event $\mathcal{E}_1\cap\mathcal{E}_2$, it must hold that $f_n$ is calibrated. To see why this is true, we split into cases:
\begin{itemize}
    \item If $\calstarhat(f_{n,\textnormal{init}}) \leq c - \frac{20+\sqrt{2\log(1/\delta)}}{\sqrt{\lfloor n/2\rfloor }}$, then our procedure defines $f_n = f_{n,\textnormal{init}}$. And, we have
    \[\calstarm(f_n) = \calstarm(f_{n,\textnormal{init}}) \leq \calstarhat(f_{n,\textnormal{init}}) + 
     \frac{20+\sqrt{2\log(1/\delta)}}{\sqrt{\lfloor n/2\rfloor}} \leq c,\]
     where the first inequality holds because we have assumed the event $\mathcal{E}_1$ holds.
     \item If instead $\calstarhat(f_{n,\textnormal{init}}) > c - \frac{20+\sqrt{2\log(1/\delta)}}{\sqrt{\lfloor n/2\rfloor }}$, then our procedure defines $f_n = f_{n,\textnormal{const}}$. And, we have
     \[\calstarm(f_n) = \calstarm(f_{n,\textnormal{const}}) = \left| \frac{1}{\lfloor n/2\rfloor}\sum_{i=\lceil n/2\rceil +1}^n  Y_i - \EE_P[Y]\right| \leq \sqrt{\frac{\log(1/\delta)}{2\lfloor n/2\rfloor}}\leq c,\]
     where the last inequality holds by assumption on $c$, the next-to-last inequality holds due to the event $\mathcal{E}_2$, and the second equality holds since, for any constant function $f(x)\equiv C$, its calibration error is given by $\calstarm(f) = |\EE[Y]-C|$.
\end{itemize}
We have therefore verified that $\calstarm(f_n)\leq c$ in both cases, and so
\begin{multline*}\PP\big\{\calstarm(f_n)\leq c \mid (X_1,Y_1),\dots,(X_{\lceil n/2\rceil},Y_{\lceil n/2\rceil})\big\} \\\geq \PP\big\{\mathcal{E}_1\cap\mathcal{E}_2\mid (X_1,Y_1),\dots,(X_{\lceil n/2\rceil},Y_{\lceil n/2\rceil})\big\}  \geq 1-2\delta.\end{multline*}

\subsection{Proof of \cref{prop:estimation}: \calstar can be estimated at the parametric rate for fixed $f$\label{proof:estimation}}

\begin{proof}[\cref{prop:estimation}] 
We proceed by turning the problem into a form for which we can apply empirical process theory. For notational brevity, we adopt the notation that $\mathbb E_n[g(X,Y)] := \frac{1}{n} \sum_{i=1}^n g(X_i, Y_i)$, where $g$ is an arbitrary function.

Observe that
\begin{align*}
    &|\calstarhat(f) -\calstarm(f)| \\
    &= \left|\sup_{I: \text{ interval}} \left|\mathbb E\left[(Y-f(X))\ind{f(X) \in I}\right]\right| - \sup_{I': \text{ interval}} \left|\mathbb E_n\left[(Y-f(X))\ind{f(X) \in I'}\right]\right|\right|\\
    &\leq \sup_{I: \text{ interval}}| |\mathbb E[(Y-f(X))\ind{f(X) \in I}]| -|\mathbb E_n[(Y-f(X))\ind{f(X) \in I}]||\\
    &\leq \sup_{I: \text{ interval}}| \mathbb E[(Y-f(X))\ind{f(X) \in I}] -\mathbb E_n[(Y-f(X))\ind{f(X) \in I}]|,
\end{align*}
where the first inequality follows from the triangle inequality of suprema.

We can now apply empirical process theory. By \citet[Theorem 2.2]{koltchinskii2011oracle}, we know that the relevant Rademacher complexity of this empirical process is bounded by twice the Rademacher complexity of 
\begin{equation}
    \{g: [0,1]\rightarrow \{0,1 \} \mid g(z) = \ind{x \in I} \text{ for some } I\subseteq[0,1] \text{ interval}  \}.\label{eq:class_indicator_variables_over_intervals}
\end{equation}
By \cref{lemma:bv_rademacher_complexity}, we know that the Rademacher complexity of \eqref{eq:class_indicator_variables_over_intervals} is at most $5/\sqrt{n}$.

Therefore, by \citet[Theorem 4.10]{wainwright2019high}, we know that $\forall \epsilon\geq 0$
\begin{equation*}
    \mathbb P \{ |\calstarhat(f) -\calstarm(f)| \leq 20/\sqrt{n} + \epsilon \} \geq 1-\exp(-n\epsilon^2/2).
\end{equation*}

Thus, $\forall \delta>0$, 
\begin{equation*}
    |\calstarhat(f) -\calstarm(f)| \leq \frac{20+\sqrt{2\log(1/\delta)}}{\sqrt{n}}
\end{equation*}
with probability at least $1-\delta$.
\end{proof}

\begin{lemma} \label{lemma:bv_rademacher_complexity}
Let $\epsilon_i,\dots, \epsilon_n$ be iid Rademacher random variables. Then,
    \begin{equation*}
        \mathbb E \left[\sup_{\substack{f \in [-1,1]^n \\ \sum_{i=1}^{n=1}|f_{i+1}-f_i|\leq M}} \frac{1}{n} \sum_{i=1}^n \epsilon_i f_i \right] \leq \frac{2M+1}{\sqrt{n}}.
    \end{equation*}
\end{lemma}
\begin{proof}
    Observe that
    \begin{align*}
        \sum_{i=1}^n \epsilon_i f_i &= f_1 \sum_{i=1}^n \epsilon_i + \sum_{i=2}^n \epsilon_i (f_i-f_1)\\
        &= f_1 \sum_{i=1}^n \epsilon_i + \sum_{i=2}^n \epsilon_i \sum_{j=1}^{i-1}(f_{j+1}-f_j)\\
        &= f_1 \sum_{i=1}^n \epsilon_i + \sum_{j=1}^{n-1} (f_{j+1}-f_j) \sum_{i=j+1}^n \epsilon_i.
    \end{align*}
    Therefore,
    \begin{equation*}
        \sup_{\substack{f \in [-1,1]^n \\ \sum_{i=1}^{n=1}|f_{i+1}-f_i|\leq M}} \sum_{i=1}^n \epsilon_i f_i \leq \left|\sum_{i=1}^n \epsilon_i\right| + M \max_{1\leq j \leq n-1} \left|\sum_{i=j+1}^n \epsilon_i\right|.
    \end{equation*}

Using an Jensen's inequality and an easy-to-show fact of random walks (see, e.g., \citet{lawler2010random}), we get
\begin{equation*}
    \mathbb E \left|\sum_{i=1}^n \epsilon_i\right| \leq \ \sqrt{ \mathbb E \left[\left(\sum_{i=1}^n \epsilon_i\right)^2\right]} = \sqrt{n}.
\end{equation*}

Using the Lévy inequality, we similarly get
\begin{align*}
    \mathbb E \left[\max_{1\leq j \leq n-1} \left|\sum_{i=j+1}^n \epsilon_i\right|\right] &= \mathbb E \left[\max_{1\leq j \leq n-1} \left|\sum_{i=1}^j \epsilon_i\right|\right]\\
    &= \int_0^\infty \mathbb P \left\{ \max_{1\leq j \leq n-1} \left|\sum_{i=1}^j \epsilon_i\right|\geq t\right\} dt\\
    &\leq 2 \int_0^\infty \mathbb P \left\{ \left|\sum_{i=1}^{n-1}\epsilon_i\right|\geq t\right\} dt\\
    &= 2 \mathbb E \left|\sum_{i=1}^{n-1}\epsilon_i\right|\\
    &= 2\sqrt{n-1}.
\end{align*}

Therefore,
\begin{equation*}
        \mathbb E \left[\sup_{\substack{f \in [-1,1]^n \\ \sum_{i=1}^{n=1}|f_{i+1}-f_i|\leq M}} \frac{1}{n} \sum_{i=1}^n \epsilon_i f_i \right] 
        \leq \frac{\sqrt{n}+2M\sqrt{n-1}}{n}
        \leq \frac{2M+1}{\sqrt{n}}.
    \end{equation*}

\end{proof}

\section{Post-hoc calibration: proofs and counter-examples\label{apdx:post_hoc}}

\subsection{Proof of \cref{prop:iso}: isotonic regression enjoys uniform asymptotic control of \calstar\label{proof:iso}}

\begin{proof}[\cref{prop:iso}]
Observe that
\begin{equation*}
    \begin{split}
    \calstarm^P(h_n \circ f) \leq &\sup_{I: \text{ interval}} |\mathbb E [Y \ind{h_n(f(X)) \in I} \mid h_n] - \mathbb E_n[h_n(f(X)) \ind{h_n(f(X)) \in I}]| \\
    &+ \sup_{I: \text{ interval}} |\mathbb E_n[h_n(f(X)) \ind{h_n(f(X)) \in I}] - \mathbb E[h_n(f(X)) \ind{h_n(f(X)) \in I}\mid h_n]|.
\end{split}
\end{equation*}

Then, using \cref{lemma:iso_piecewise_empirical_mean},
\begin{align}\label{eq:iso_proof_reference_platt_scaling}
    \calstarm^P(h_n \circ f) \leq &\sup_{I: \text{ interval}} |\mathbb E [Y \ind{h_n(f(X)) \in I} \mid h_n] - \mathbb E_n[Y \ind{h_n(f(X)) \in I}]| \\
    &+ \sup_{I: \text{ interval}} |\mathbb E_n[h_n(f(X)) \ind{h_n(f(X)) \in I}] - \mathbb E[h_n(f(X)) \ind{h_n(f(X)) \in I}\mid h_n]|.\notag
\end{align}

Since $h_n$ is monotone, we know that if $I \subseteq [0,1]$ is an interval, then $h_n^{-1}(I)$ must also be an interval. Therefore,
\begin{equation*}
    \begin{split}
    \calstarm^P(h_n \circ f) \leq &\sup_{I: \text{ interval}} |\mathbb E [Y \ind{f(X) \in I} \mid h_n] - \mathbb E_n[Y \ind{f(X) \in I}]| \\
    &+ \sup_{I: \text{ interval}} |\mathbb E_n[h_n(f(X)) \ind{f(X) \in I}] - \mathbb E[h_n(f(X)) \ind{f(X) \in I}\mid h_n]|.
\end{split}
\end{equation*}

Then, using \cref{lemma:unimodal_interval}, we get
\begin{equation*}
    \begin{split}
    \calstarm^P(h_n \circ f) \leq &\sup_{I: \text{ interval}} |\mathbb E [Y \ind{f(X) \in I} \mid h_n] - \mathbb E_n[Y \ind{f(X) \in I}]| \\
    &+ \sup_{I: \text{ interval}} |\mathbb E_n[ \ind{f(X) \in I}] - \mathbb E[ \ind{f(X) \in I}\mid h_n]|.
\end{split}
\end{equation*}

Both terms no longer depends on $h_n$, so we can rewrite this as
\begin{equation*}
    \begin{split}
    \calstarm^P(h_n \circ f) \leq &\sup_{I: \text{ interval}} |\mathbb E [Y \ind{f(X) \in I}] - \mathbb E_n[Y \ind{f(X) \in I}]| \\
    &+ \sup_{I: \text{ interval}} |\mathbb E_n[ \ind{f(X) \in I}] - \mathbb E[ \ind{f(X) \in I}]|.
\end{split}
\end{equation*}

Using a similar Rademacher complexity argument from \cref{prop:estimation}, we know with probability at least $1-\delta/2$ that each term individually is bounded above by $\frac{20+\sqrt{2\log(2/\delta)}}{\sqrt{n}}$ and $\frac{10+\sqrt{2\log(2/\delta)}}{\sqrt{n}}$, respectively. Note that the high-probability bound for the first term is higher than the second due to the contraction mapping argument doubling the Rademacher complexity.

Therefore, by the union bound, we know that 
\begin{equation*}
    \calstarm^P(h_n \circ f) \leq \frac{30+2\sqrt{2\log(2/\delta)}}{\sqrt{n}},
\end{equation*}
with probability at least $1-\delta$.
\end{proof}

\begin{lemma}\label{lemma:iso_piecewise_empirical_mean}
    Let $\hat f$ refer to running isotonic regression on $\{(X_i, Y_i) \}_{i \in [n]}$, where $Y \in \mathbb R$. Then, for any $A \subseteq \mathbb R$,
    \begin{equation*}
        \frac{1}{n} \sum_{i=1}^n \hat{f}(X_i) \ind{\hat{f}(X_i) \in A} = \frac{1}{n} \sum_{i=1}^n Y_i \ind{\hat{f}(X_i) \in A}.
    \end{equation*}
\end{lemma}

\begin{proof}
    This follows immediately from the properties of isotonic regression: namely, it is a piecewise empirical mean.
\end{proof}

\begin{lemma}\label{lemma:unimodal_interval}
    For any unimodal $g: \mathbb R \rightarrow [0,1]$ and any distributions $P,Q$ on $\mathbb R$,
    \begin{equation*}
        |\mathbb E_P[g(X)] - \mathbb E_Q[g(X)]| \leq \sup_{\text{Interval } I \subseteq \mathbb R} |P(I) - Q(I)|.
    \end{equation*}
\end{lemma}

\begin{proof}
    Observe that
    \begin{align*}
        |\mathbb E_P[g(X)] - \mathbb E_Q[g(X)]| &= \left|\int_{\mathbb R} g(x) (dP - dQ)(x)\right|\\
        &= \left|\int_{\mathbb R} \int_0^1 \ind{t \leq g(x)} dt (dP - dQ)(x)\right|\\
        &= \left|\int_0^1 \int_{\mathbb R} \ind{g(x) \geq t} (dP - dQ)(x) dt\right|\\
        &\leq \sup_{t \in [0,1]}\left|\int_{\mathbb R} \ind{g(x)\geq t} (dP - dQ)(x)\right|.
    \end{align*}

    Let 
    \begin{equation*}
        I := \{x: g(x) \geq t \}.
    \end{equation*}
    Since $g$ is unimodal, we know that $I$ must be an interval.

    Therefore,
    \begin{equation*}
        \sup_{t \in [0,1]}\left|\int_{\mathbb R} \ind{g(x)\geq t} (dP - dQ)(x)\right| \leq \sup_{\text{Interval } I \subseteq \mathbb R} |P(I) - Q(I)|.
    \end{equation*}
\end{proof}

\subsection{Modifying Platt scaling to achieve Cutoff calibration\label{sec:modified_platt_scaling}}

In this section, we propose a modified Platt scaling procedure that is able to achieve asymptotic control of \calstar by testing for model misspecification. This may be beneficial to users who prefer Platt scaling due to its empirical success in the limited sample regime \citep{niculescu2005predicting}.

Define the usual Platt-scaling post-processor as $\Tilde{h}_n$. Then, our final $h_n$ is designed to ensure $\calstarhat(h_n \circ f)$ is small, up to a user-specified threshold $\epsilon_n$:
\begin{equation*}
    h_n(z) = \begin{cases}
         \Tilde{h}_n(z) & \text{if } \calstarhat(\Tilde{h}_n \circ f) \leq \epsilon_n\\
         \mathbb E_n [Y] & \text{otherwise}
    \end{cases}\label{eq:modified_platt_scaling}.
\end{equation*}

We can ensure that this procedure achieves a similar asymptotic guarantee to our isotonic regression result.

\begin{proposition}    \label{prop:modified_platt_scaling}
    Suppose $(X_i, Y_i) \overset{iid}{\sim}P$ for $i=1,\dots, n$. Assume $Y_i \in [0,1]$ and $f$ is fixed. Let $h_n$ be as defined in \eqref{eq:modified_platt_scaling}. Fix $\delta>0$. Let $\epsilon_n>0$ be a user-specified tolerance for model-misspecification. Then, with probability at least $1-\delta$,
    \begin{equation*}
        \calstarm^P(h_n \circ f) \leq \frac{30+2\sqrt{2\log(2/\delta)}}{\sqrt{n}} + \epsilon_n.
    \end{equation*}
\end{proposition}

We note that it may be natural to set $\epsilon_n = \frac{20+\sqrt{2\log(20)}}{\sqrt{n}}$, which corresponds to the 95\% upper confidence bound of $\calstarhat(f)$ when $\calstarm(f) = 0$.

\begin{proof}
Observe that returning the unconditional empirical average of $Y$ ensures that $\calstarhat$ is 0:
\begin{equation*}
    \calstarhat(\mathbb E_n[Y]) = 0.
\end{equation*}

Therefore, we know
\begin{equation*}
    \calstarhat(h_n \circ f) \leq \epsilon_n.
\end{equation*}

Then,
\begin{align*}
    &\calstarm^P(h_n \circ f)\\
    &= \calstarm^P(h_n \circ f) - \calstarhat(h_n \circ f) + \calstarhat(h_n \circ f)\\
    &\leq \calstarm^P(h_n \circ f) - \calstarhat(h_n \circ f) + \epsilon_n\\
    &\leq \sup_{I: \text{ interval}} \left|\mathbb E [(Y-h_n \circ f(X))\ind{h_n \circ f(X) \in I}] - \frac{1}{n} \sum_{i=1}^n (Y_i-h_n \circ f(X_i)) \ind{h_n \circ f(X_i) \in I}\right| + \epsilon_n.
\end{align*}
The first term can be shown to be $\lesssim \frac{1}{\sqrt{n}}$ with probability at least $1-\delta$ following the same arguments in the isotonic regression proof, from equation \eqref{eq:iso_proof_reference_platt_scaling} onward, since $h_n$ is also monotone here.
\end{proof}

\end{document}